\documentclass[11pt,journal,draftcls,onecolumn]{IEEEtran}

\usepackage{geometry}
\usepackage[utf8]{inputenc} 
\usepackage[T1]{fontenc}

\ifCLASSINFOpdf

\else

\fi
\usepackage[cmex10]{amsmath}
\usepackage{ifthen}
\usepackage{tikz}
\usepackage{cite}
\usepackage{amsthm,amsfonts,amssymb}
\usepackage{graphicx}
\usepackage{subfig}
\usepackage{blkarray}
\usepackage{dsfont}
\usepackage[mathscr]{euscript}
\usepackage{enumitem}
\usepackage{algorithm}
\usepackage[noend]{algpseudocode}
\usepackage{blkarray}
\usepackage{stfloats}%
\newtheorem{theorem}{Theorem}[section]
\newtheorem{corollary}{Corollary}[section]
\newtheorem{proposition}{Proposition}[section]
\newtheorem{lemma}{Lemma}[section]
\theoremstyle{remark}
\newtheorem*{remark}{Remark}
\theoremstyle{definition}

\usepackage{url}
\usepackage{tcolorbox}
\usepackage{arydshln}
\usepackage{mathtools}
\usepackage{dsfont}

\definecolor{brickred}{cmyk}{0,0.89,0.94,0.28}
\definecolor{goldenrod}{cmyk}{0,0.10,0.84,0}
\definecolor{purple}{cmyk}{0.45,0.86,0,0}
\definecolor{rawsienna}{cmyk}{0,0.72,1,0.45}
\definecolor{olivegreen}{cmyk}{0.64,0,0.95,0.40}
\definecolor{peach}{cmyk}{0,0.5,0.7,0}
\definecolor{darkolive}{rgb}{0.,0.4,0.}
\colorlet{grey}{gray!40}

\global\long\def\P{\mathbb{P}}
\global\long\def\E{\mathbb{E}}

\global\long\def\d{\mathrm{d}}

\DeclareMathOperator{\Del}{Del}

\DeclareMathOperator{\Ber}{Ber}
\usepackage{graphicx}
\makeatletter
\newcommand{\ostar}{\mathbin{\mathpalette\make@circled\star}}
\newcommand{\make@circled}[2]{%
	\ooalign{$\m@th#1\smallbigcirc{#1}$\cr\hidewidth$\m@th#1#2$\hidewidth\cr}%
}
\newcommand{\smallbigcirc}[1]{%
	\vcenter{\hbox{\scalebox{0.77778}{$\m@th#1\bigcirc$}}}%
}
\makeatother

\usepackage[cmintegrals]{newtxmath}

\begin{document}

\title{Information Rates Over Multi-View Channels}
\author{
	V.~Arvind~Rameshwar\ 
	and\ 
	Nir Weinberger
	\thanks{V.~A.~Rameshwar is with India Urban Data Exchange, SID, Indian Institute of Science, Bengaluru 560094, email: arvind.rameshwar@gmail.com. N.~Weinberger is with the Department of Electrical and Computer Engineering, Technion, Haifa 3200003, Israel,
		email: nirwein@technion.ac.il . This paper was accepted in part to the 2024 IEEE International Symposium on Information Theory. The research of N. W. was supported by the Israel Science Foundation (ISF), grant no. 1782/22.}
}
\IEEEoverridecommandlockouts
%


\maketitle

\begin{abstract}
	We investigate the fundamental limits of reliable communication over multi-view channels, in which the channel output is comprised of a large number of independent noisy views of a transmitted symbol. We consider first the setting of multi-view discrete memoryless channels and then extend our results to general multi-view channels (using multi-letter formulas). We argue that the channel capacity and dispersion of such multi-view channels converge exponentially fast in the number of views to the entropy and varentropy of the input distribution, respectively. We identify the exact rate of convergence as the smallest Chernoff information between two conditional distributions of the output, conditioned on unequal inputs. For the special case of the deletion channel, we compute upper bounds on this Chernoff information. Finally, we present a new channel model we term the Poisson approximation channel --- of possible independent interest --- whose capacity closely approximates the capacity of the multi-view binary symmetric channel for any fixed number of views.
\end{abstract}



\IEEEpeerreviewmaketitle
\section{Introduction}
Consider a communication setting, in which the decoder receives multiple, potentially noisy, views of a single transmitted sequence, which is sent via independent transmissions through a stochastic noise channel. Clearly, multiple noisy views (or multiple draws) of a transmitted sequence can only improve the error probabilities of decoding, as compared to decoding a single received sequence. Viewed differently, it is also clear that the Shannon capacity of the \emph{multi-view} channel, as described above, is at least as large as the capacity of the standard single-view channel.

In this paper, we consider first the setting where the noisy channel is a discrete memoryless channel (DMC) and then extend our results to general multi-letter channels. Multi-view DMCs were first studied in the classic work by Levenshtein \cite{levenshtein} on sequence reconstruction. Among several other results, \cite{levenshtein} presented bounds on the probability of error under ML decoding and characterized the number of noisy views required to guarantee recovery of the transmitted sequence, with \emph{decaying error probability}. These results are determined by the minimum Chernoff information between conditional distributions of outputs of the DMC, conditioned on unequal inputs. Levenshtein's result on the number of noisy views required for reconstruction, however, depends crucially on the assumption that the error probability decays with the blocklength, and leaves open the question of characterizing the rates achievable over DMCs with fixed error probability, using finitely many channel uses. Furthermore, the proof strategy in \cite{levenshtein} resists extension to channels, such as those that model synchronization errors, wherein the lengths of the input and output sequences could be different. 

In this work, we characterize exactly the rate of convergence of the information rate and the channel dispersion of general multi-view channels, under arbitrary input distributions, to the input entropy and the varentropy, respectively, when the number of views is large. As our main result, we show that the rates of (exponential) convergence are again determined by the minimum Chernoff information between conditional distributions, as above.  This result  does not assume any additional structure on the channel such as symmetry, or binary alphabets (though require the input and output alphabet to be finite)

While our result on the convergence of information rate to the input entropy can be derived from the results in \cite{levenshtein} (and \cite[Thm. 5]{shgalber})\footnote{We show, moreover, that for binary-input memoryless symmetric (BIMS) channels with a uniform input distribution, it is possible to derive the convergence of the information rate (which, in this case, is the capacity of the channel) to the input entropy, using a simple argument from binary hypothesis testing.}, we present new tools for deriving such a result, which also helps us obtain the rate of convergence of the channel dispersion to the input varentropy.  An important consequence of our results for DMCs (under some regularity conditions) is our characterization of the largest rates achievable over finitely many channel uses, with any \emph{fixed error probability}, using results from finite-blocklength information theory, thereby resolving a natural question that stems from Levenshtein's work \cite{levenshtein}.  We then show that our results also hold for multi-letter channels with potentially different lengths of input and output sequences. Furthermore, for the special case of the (multi-letter) deletion channel, which is a popular synchronization error channel \cite{mahdi_survey}, we compute explicit upper bounds on the above Chernoff information, and discuss some interesting consequences.

Our interest in this setting is strongly motivated by the recently proposed model for DNA-based storage systems (see, e.g., \cite{dna1,shomorony,nir_coded,nir_merhav} and \cite{coding_sets1,coding_sets2} for coding schemes). In a nutshell, in this channel, information is encoded into a large multiset of short sequences from a quarternary alphabet, called molecules (or strands). These sequences are then stored in a pool without order, and reading the information is performed by sampling molecules from the pool. After being inserted into the pool, the molecules are ``amplified'' by a process called Polymerase Chain Reaction (PCR), to produce several copies of each molecule. Therefore the sampling process during reading can be accurately modelled by sampling \textit{with} replacement. Each sampled molecule is then sequenced to read the written symbols it consist of. This sequencing operation is noisy, and the noise is independent from one molecule read to the other. Furthermore, in some cases, this noisy operation is well modelled by a DMC, which tends to be \textit{asymmetric} \cite{asymlee}.

Evidently, in this DNA-storage model, each molecule is possibly read many times, and so the multi-view mutual information between the molecule sequence and its multiple output reads plays a pivotal role in determining the capacity of this storage system. Indeed, it was shown in  \cite[Thm. 5]{nir_merhav}, that the capacity is composed of a weighted average of such mutual information terms, which accounts for the multiple input molecules, and the possibility that each of them is sampled a different number of times. 

The total number of sampled molecules in the reading operation is a design parameter, say the number  of moelcueles sampled may be a constant factor of the total number of molecules, say $d_0$. In this case, it is expected that each input molecule is also sampled about $d_0$ times, and the mutual information term between the input molecule sequences and its multiple output reads falls into the setting considered in this paper. Moreover, as said, each molecule on its own is rather short, and so finite-blocklength bounds appears to be relevant. Specifically, given a molecule length, and $d_0$, one may determine using our results when the maximal multi-view rate of this molecule is approached, for a given error probability. A detailed investigation of this setting is beyond the scope of this paper, but of interest for future work. 

Previous works considered the mutual information  for a fixed number of views. In this non-asymptotic setting, the capacity of a multi-view binary symmetric channel (BSC) was identified and an explicit expression was provided in \cite{mitzenmacher_datarecovery}. More generally, bounds on the capacities of multi-view DMCs were derived in the literature on ``information combining'' (see \cite{landhuber,extremes} and the references therein). Also related is the well-studied problem on \textit{trace reconstruction} \cite{trace1,trace2,trace3}, wherein an input string is passed independently through a deletion channel and the goal is to find the  number of resultant outputs required for reconstruction with high-probability. By contrast, in this work we focus on the scaling to  mutual information and  dispersion in the asymptotic regime of large number of views. 

As an additional contribution, we also explore the non-asymptotic capacity of the so-called \emph{Poisson approximation channel}, which we believe is of possible independent interest. We show that the capacity of this channel is a tight lower bound on the capacity of the multi-view BSC. Such a technique also allows us to obtain simple bounds on the non-asymptotic capacities of general binary-input memoryless symmetric (BIMS) channels.

The rest of the paper is organized as follows: Section \ref{sec:notation} sets down the notation; Section \ref{sec:main} presents our main results for DMCs; Section \ref{sec:convergence} provides a proof sketch of the rate of convergence of the information rate and channel dispersion of a DMC to their respective limits. Section \ref{sec:deletion} then extends our results to general multi-letter channels and presents upper bounds on Chernoff information discussed above for the deletion channel. Section \ref{sec:poisson} contains a proof that the capacity of the Poisson approximation channel is a tight lower bound on the capacity of the multi-view BSC. Section \ref{sec:proofs} then provides full proofs of our main convergence results. Section \ref{sec:conclusion} concludes the paper and proposes problems for further research.
\section{Notation Conventions}
\label{sec:notation}
All logarithms are assumed to be natural logarithms. Random variables are denoted by capital letters, e.g., $X, Y$, and small letters, e.g., $x, y$, denote their instantiations. Sets are denoted by calligraphic letters, e.g., $\mathcal{X}, \mathcal{Y}$. The set of positive natural numbers $\{1,2,\ldots\}$ is denoted as $\mathbb{N}$. Notation such as $P(x), P(y|x)$ are used to denote the probabilities $P_X(x), P_{Y|X}(y|x)$, when it is clear which random variables are being referred to. For $\alpha\in [0,1]$ and $n\geq 1$,  Ber$(\alpha)$ and Bin$(n,\alpha)$,  denote, respectively, the Bernoulli distribution, supported on $\{0,1\}$, with $\Pr_{X\sim \text{Ber}(\alpha)}[X=1] = \alpha$, and the Binomial distribution, supported on $\{0,1,\ldots,n\}$, with 

\begin{equation}
\Pr_{X\sim \text{Bin}(n,\alpha)}[X=k] = \binom{n}{k}\alpha^k (1-\alpha)^{n-k},
\end{equation}
for $k\in \{0,1,\ldots,n\}$. 

The entropy of $X\sim P_X$ is denoted by $\mathsf{H}(X) =\mathsf{H}(P_X):=  \mathbb{E}[-\log P_X(X)]$ and its varentropy is denoted by
\begin{equation}
\mathsf{V}(X) := \mathbb{E}\left[\left(-\log {P_X(X)} - \mathsf{H}(X)\right)^2\right].
\end{equation}
Both only depend only on the distribution of $X$. For $\gamma\in [0,1]$, the binary entropy function is denoted by $h_b(\gamma):= -\gamma\log \gamma - (1-\gamma)\log (1-\gamma)$. The Chernoff information $\mathsf{C}(P,Q)$ between two distributions $P$ and $Q$ defined on the same alphabet $\mathcal{X}$ is given by (see \cite[Ch. 11]{cover_thomas}) 
\begin{equation}\mathsf{C}(P,Q) := -\min_{\lambda\in [0,1]}\log \left(\sum_{x\in \mathcal{X}} P(x)^{1-\lambda}Q(x)^\lambda\right). \label{eq:chernoff}\end{equation}
Furthermore, the standard Kullback-Leibler (KL) divergence between two distributions $P,Q$ on the same alphabet $\mathcal{X}$ is given by 
\begin{equation}
D(P||Q) := \E_{X\sim P}\left[\log \frac{P(X)}{Q(X)}\right],
\end{equation}
where we take $0\log \frac{0}{z} = 0$ and we set $D(P||Q) = \infty$ if there exists $x\in \mathcal{X}$ such that $Q(x) = 0$ but $P(x)>0$.
The the length-$n$ vector $(a_1,\ldots,a_n)$ is denoted by $a^n$, and the Hamming weight, i.e.,  the number of ones in $a^n$, is denoted by $w(a^n)$.  
The indicator that $x$ equals $b$ is denoted by $\mathds{1}\{x = b\}$; this indicator equals $1$ if $x = b$ and $0$, otherwise. For sequences $(a_n)_{n\geq 1}$ and $(b_n)_{n\geq 1}$ of positive reals, $a_n = o(b_n)$ if $\lim_{n\to \infty} \frac{a_n}{b_n} = 0$, $a_n = O(b_n)$ if for all $n$ large enough, $a_n\leq c\cdot b_n$, for some $c>0$, and $a_n = \Theta(b_n)$ if for all $n$ large enough, we have that $m\cdot b_n\leq a_n\leq M\cdot b_n$, for some constants $m,M\in \mathbb{R}$. 

\section{Main Results for DMCs}
\label{sec:main}
We now formalize the problem for the case  the noisy channel is a DMC, and present our main results. Later on, in Section \ref{sec:deletion} we will extend Theorem \ref{thm:conv} here to general multi-letter channels. Consider a DMC $W$ with finite input and output alphabets $\mathcal{X}$ and $\mathcal{Y}$, respectively, specified by the channel law $\{P_{Y|X}(y|x) = W(y|x):\ x\in \mathcal{X}, y\in \mathcal{Y}\}$. Let $X\sim P_X$ be the input to the channel. Our setting of interest is when $d\geq 1$ noisy outputs are obtained by independently passing the input through the DMC. Let these outputs be denoted by $Y_1,\ldots, Y_d$; their conditional distribution obeys
\begin{equation}
	\label{eq:d-view}
P(Y^d = y^d\mid X = x) = \prod_{i=1}^{d} W(y_i|x).
\end{equation}
We call a DMC with a channel law as above as the \emph{$d$-view} DMC $W^{(d)}$. We assume further that $|\mathcal{X}|$ and $|\mathcal{Y}|$ do not depend on $d$. 

Now, given any DMC $W$, we denote its capacity by \cite{Sh48} \begin{equation}C(W) = \max_{P_X} I(X;Y) = \max_{P_X} \E\left[\log \frac{W(Y|X)}{\sum_{x\in \mathcal{X}} P(x)W(Y|x)}\right].\end{equation} Hence, the capacity of $W^{(d)}$ is given by $C^{(d)}:= \max_{P_X} I^{(d)},$ where $I^{(d)}:= I(X;Y^d)$. Further, for any input distribution $P_X$, the channel dispersion of the $d$-view DMC $W$ is given by 
\begin{equation}
V^{(d)}:= \mathbb{E}\left[\left(\iota(X;Y^d) - I(X;Y^d)\right)^2\right],
\end{equation}
where 
\begin{equation}
\iota(x;y^d):= \log \frac{P(y^d|x)}{P(y^d)}
\end{equation}
is the \textit{information density}. 
We are interested in characterizing the behaviour of $C^{(d)}$ (or more generally, of $I^{(d)}$) and $V^{(d)}$ when $d$ is large.  From operational considerations, we can argue that $I^{(d)}$ and $V^{(d)}$ are non-decreasing in $d$. More formally, for $I^{(d)}$, the data processing inequality \cite[Thm. 2.8.1]{cover_thomas} guarantees that $I^{(d)}$ is non-decreasing in $d$. Furthermore, it is reasonable to expect that $I^{(d)}$ and $V^{(d)}$ converge to $\mathsf{H}(X)$ and $\mathsf{V}(X)$, respectively, since for very large $d$, the ``uncertainty'' introduced by the channel is expected to drop to $0$. In this work, we confirm these convergence results by showing that the $I^{(d)}$ and $V^{(d)}$ in fact converge to their limiting values exponentially quickly in $d$, and explicitly identify the rate of convergence.

To build some intuition why the convergence is exponentially fast in $d$, at least for $I^{(d)}$, consider the simple example of the $d$-view BSC$(p)$, which we denote by BSC$^{(d)}(p)$ (here, $p<1/2$ is the crossover probability of the BSC), with a uniform input distribution. At the decoder end, we construct a simple ``data processor'' that takes the (random) outputs $Y^d\in \{0,1\}^d$ and computes their majority $M\in \{0,1\}$.  Intuitively, the more the number of noisy output sequences (or views) that are available to the decoder, the better the estimate that results from majority decoding. More precisely, we have 
\begin{align}
	I^{(d)} &= I(X;Y^d)\\
	&\geq I(X;M) = \log 2-h_b(\Pr[X\neq M]), 
\end{align}
where the inequality holds by the data processing inequality and the last equality holds by symmetry. Now, observe that since $p<1/2$, we have by the Chernoff bound (see \cite[Example 1.6.2]{gallager_stoch}) that
\begin{align}
\Pr[X\neq M]&\leq \text{exp}(-d\cdot D(\text{Ber}(1/2)||\text{Ber}(p)))\\
&= \text{exp}(-d\cdot Z(p)),
\end{align}
where $Z(p) = 2\sqrt{p(1-p)}$ is the Bhattacharya parameter for the BSC$(p)$. Therefore, for large enough $d$, we have that $\Pr[X\neq M]$ is sufficiently small so that \begin{equation}h_b(\Pr[X\neq M])\leq -2\Pr[X\neq M]\cdot \log \Pr[X\neq M].\end{equation} Hence, for sufficiently large $d$, we see that $I^{(d)}\geq \log 2-dZ(p)\cdot \text{exp}(-d\cdot Z(p))$, thereby providing reason for the intuition that the convergence of $I^{(d)}$ to $\log 2$ is exponentially fast in $d$. In fact, in Theorem \ref{thm:conv}, we argue that the rate $Z(p)$ is in fact tight for the BSC$(p)$.


Our main theorem is given below.
\begin{theorem}
	\label{thm:conv}
	We have that 
\begin{equation}
I^{(d)} = \mathsf{H}(X)- \textnormal{exp}\left({-d \rho +\Theta(\log d|\mathcal{X}|)}\right),\ 
\end{equation}
and
\begin{equation}
\left\lvert V^{(d)} - \mathsf{V}(X)\right \rvert= \textnormal{exp}\left({-d \rho +\Theta(\log d|\mathcal{X}|)}\right),
\end{equation}
where 
 \begin{equation}
 \rho = \min_{x,x': x\neq x'} \mathsf{C}(P_{Y|x},P_{Y|x'}).
 \end{equation}
\end{theorem}
Note that the theorem above captures the dependence of the speed of convergence of mutual information and dispersion to their limits in terms of both the number of views $d$ and the size of the input alphabet (assumed to be independent of $d$). The dependence on $|\mathcal{X}|$ will assume significance when we transition to multi-letter channels in Section \ref{sec:deletion}. 

We now proceed to interpret Theorem \ref{thm:conv}. Suppose we define the rates \begin{align}
		\rho(I)&:= \liminf_d -\frac{1}{d}\cdot \log \left(\mathsf{H}(X) - I^{(d)}\right)\\
		&= \liminf_d -\frac{1}{d}\cdot \log \mathsf{H}(X|Y^d),
	\end{align}
and
\begin{align}
		\rho(V)&:= \liminf_d -\frac{1}{d}\cdot \log \left\lvert\mathsf{V}(X) - V^{(d)}\right\rvert,
	\end{align} we get by the above result that these rates of convergence obey $\rho(I) = \rho(V) = \rho$, which in turn is independent of the input distribution $P_X$, with possible dependence on its support. Further, the limits in the definitions of these rates exist. Moreover, note that for a binary-input memoryless symmetric (BIMS) channel $W$ (see \cite[Ch. 4]{mct} for a definition), the function \begin{equation}f(\lambda):=\log \left(\sum_{y\in \mathcal{Y}} P(y|0)^{1-\lambda}P(y|1)^\lambda\right)\end{equation} is symmetric about $\lambda = 1/2$. Hence, using the fact that $f$ is convex, we obtain that the minimum in \eqref{eq:chernoff} is attained at $\lambda = 1/2$, giving
	\begin{align}
	\mathsf{C}(P_{Y|0},P_{Y|1}) &= -\log \sum_{z\in \mathcal{Z}} \sqrt{P_{Y|0}(z|0)P_{Y|1}(z|1)}\\ &= -\log Z_b(W),
	\end{align}
	where $Z_b(W)$ is the Bhattacharya parameter of the BIMS channel $W$. Hence, following the results in \cite{extremes_error}, the rates of convergence for a BIMS channel $W$ are extremized by a BEC and BSC having the same capacity as $W$, respectively, with the BEC having the largest and BSC having the smallest Bhattacharya parameters, respectively. Here, for the $d$-view BSC$(p)$, we have $\rho(I) = \rho(V) = -\log 2\sqrt{p(1-p)}$, and for the $d$-view BEC$(\epsilon)$, we have $\rho(I) = \rho(V) = -\log \epsilon$. Furthermore, for the Z-channel with parameter $\delta$, $\rho(I) = \rho(V) = -\log \delta$. In fact, for most well-behaved DMCs $W$, we will have that $\rho(I) = \rho(V)>0$, implying exponential convergence of mutual information and dispersion to $\mathsf{H}(X)$ and $\mathsf{V}(X)$, respectively.

Theorem \ref{thm:conv} along with the normal approximation for finite blocklengths \cite{finiteblock} then allow us to characterize, up to exponential tightness, the largest finite-blocklength rates achievable over a $d$-view DMC $W^{(d)}$, under some regularity conditions. Our presentation follows that in \cite{tanfnt} (see also \cite[Thm. 22.2]{polyanskiywu}). Suppose that the DMC $W^{(d)}$ is non-singular (see the definition in \cite[Sec. 4.2.1]{tanfnt}). Let $M^\star(n,\epsilon)$ be the largest integer $M$ such that there exists a blocklength-$n$ code with maximal probability of error $\epsilon>0$ over $W^{(d)}$. Further, suppose that under the uniform input distribution, we have that $V^{(d)}>0$.
For $t\in \mathbb{R}$, let $\Phi(t)$ be the Gaussian cumulative distribution function (c.d.f.), i.e., 
\begin{equation}\Phi(t) = \int_{-\infty}^{t} \frac{1}{\sqrt{2\pi}} e^{-z^2/2}\d z,
\end{equation}
with $\Phi^{-1}$ denoting its inverse. We then have the following theorem:
\begin{theorem}
	\label{thm:finite}
	For any non-singular DMC $W^{(d)}$, for all $n\geq 1$ and $\epsilon\in (0,1)$, we have 
	\begin{align}
		\frac{\log M^\star(n,\epsilon)}{n}&\geq \log |\mathcal{X}|-e^{-d\rho+\Theta(\log d|\mathcal{X}|)}+\Phi^{-1}(\epsilon)\cdot\frac{e^{-d\rho/2+\Theta(\log d|\mathcal{X}|)}}{\sqrt{n}}+\frac{\log n}{2n}+O\left(1\right),
	\end{align}
	where $\rho = \rho(I) = \rho(V)$.

%
%
\end{theorem}
\begin{proof}
	The proof follows from \cite[Thm. 4.1]{tanfnt} by choosing the input distribution $P_X$ to be the uniform distribution.
\end{proof}
Clearly, it follows from the above theorem that for any finite $n$, by choosing $d = \rho^{-1}\log n$, we can obtain rates \begin{equation}\frac{\log M^\star(n,\epsilon)}{n} \geq \log |\mathcal{X}|-O\left(\frac{\log n}{n}\right)\end{equation} over any $d$-view DMC $W^{(d)}$.

Next, we consider the special case of the $d$-view BSC$(p)$, which we denote by $\text{BSC}^{(d)}(p)$; the capacity of this channel is known to be equal to the capacity of the so-called binomial channel \cite[Sec. III.A]{mitzenmacher_datarecovery}, denoted by $\mathsf{Bin}_d(p)$. This capacity $C(\mathsf{Bin}_d(p))$ obeys
\begin{align}
C(\mathsf{Bin}_d(p)) = \log 2  + \sum_{i=0}^d \binom{d}{i}\left(p^i(1-p)^{d-i}\log \frac{p^i(1-p)^{d-i}}{p^i(1-p)^{d-i}+p^{d-i}(1-p)^i}\right), \label{eq:bin}
\end{align}
and the capacity is achieved by the input distribution $P_X = \text{Ber}(1/2)$. 

While Theorem \ref{thm:conv} determines the rate of convergence of $C\left(\text{BSC}^{(d)}(p)\right)=C(\mathsf{Bin}_d(p))$ (via the rate of convergence of the mutual information $I(X;Y^d)$ under a uniform input distribution $P_X = \text{Ber}(1/2)$) to $\mathsf{H}(P_X) = 1$ in the asymptotic regime of large $d$, it will be useful if a more fine-grained, non-asymptotic understanding of the capacity of this basic multi-view channel is obtainable. To this end, we identify a closely-related DMC, which we call the \emph{Poisson approximation channel} and denote by $\mathsf{Poi}_d(p)$, whose capacity turns out to be a good approximation to $C(\mathsf{Bin}_d(p))$. The DMC $\mathsf{Poi}_d(p)$ has input alphabet $\mathcal{X} = \{0,1\}$ and output alphabet $\mathcal{R} = (\mathbb{N}\cup \{0\})\times (\mathbb{N}\cup \{0\})$. We denote the input bit as $X$ and the output as the pair $(R_1,R_2)$. The channel law obeys
\begin{align}
	&P_{R_1,R_2|X}(r_1,r_2|x)= P_{R_1|X}(r_1|x)P_{R_2|x}(r_2|x), \text{for all}\ x\in \{0,1\}\ \text{and}\ (r_1,r_2)\in (\mathbb{N}\cup \{0\})\times (\mathbb{N}\cup \{0\}),  
\end{align}
where
\begin{align}
P_{R_1|0} &= \text{Poi}(d(1-p)),\ P_{R_2|0} = \text{Poi}(dp),\ \text{and}\ 
P_{R_1|1} = \text{Poi}(dp),\ P_{R_2|1} = \text{Poi}(d(1-p)).
\end{align}
Here, the random variables $R_1$ and $R_2$ are approximations for the number of $0$s and $1$s, respectively, that are received as part of the $d$ outputs of the binomial channel. Further, Poi$(\gamma)$ denotes the Poisson distribution with parameter $\gamma$. Note that $\mathsf{Poi}_d(p)$ is a BIMS channel; hence, its capacity is achieved by $P_X = \text{Ber}(1/2)$ (see  \cite[Problem 4.8]{mct}). The following theorem then holds.
\begin{theorem}
	\label{thm:poibin}
	We have that
	\begin{equation}
	 C(\mathsf{Poi}_d(p))\leq C(\mathsf{Bin}_d(p)) \leq C(\mathsf{Poi}_d(p))+\exp(-d(1-Z(p)))-Z(p)^{2d},
	\end{equation}
where $Z(p) = 2\sqrt{p(1-p)}$ is the Bhattacharya parameter of the BSC$(p)$.
\end{theorem}
Figure \ref{fig:poibin} plots the capacities of the $\mathsf{Bin}(p)$ and $\mathsf{Poi}(p)$ channels for varying values of $d$. It is clear from the plots and from Theorem \ref{thm:poibin} that $C(\mathsf{Poi}_d(p))$ is a tight lower bound on $C(\mathsf{Bin}_d(p))$; furthermore, the tightness of the lower bound improves as $d$ increases.

The result of Theorem \ref{thm:poibin} can be easily generalized to $d$-view BIMS channels $W^{(d)}$ with a finite output alphabet. It is well-known (see, e.g., Section II in \cite{boundsinfo}) that any BIMS channel $W$ with a finite output alphabet can be decomposed into finitely many BSC subchannels, i.e.,
\begin{equation}
W_{Y|X=x} = \sum_{i=1}^{K} \epsilon_i\cdot W^{(i)}_{Y|X=x},
\end{equation}
for some $K<\infty$, for any $x\in \{0,1\}$, with $W^{(i)}_{Y|X=x}$ being the channel law of a BSC with crossover probability $p_i$. Here, $\epsilon_i$ is the probability of choosing subchannel $i$, with $\sum_i \epsilon_i = 1$. Let $P_s$ be the distribution on $\{1,\ldots,K\}$ with mass $\epsilon_i$ at point $i$. We then have that
\begin{equation}
	\label{eq:bimspoi}
C^{(d)}(W) = \E\left[C(J_1,\ldots,J_d)\right],
\end{equation}
where the random variables $J_i$, $1\leq i\leq K$, are drawn i.i.d. according to $P_s$, and the notation $C(j_1,\ldots,j_d)$ refers to the capacity of that channel with the channel law
\begin{equation}
P(Y^d=y^d|X=x) = \prod_{\ell=1}^dW^{(j_\ell)}(y_\ell|x),
\end{equation}
for $x\in \{0,1\}$ and $y^d\in \mathcal{Y}^d$. Now, suppose that $0<p_1\leq p_2\leq\ldots\leq p_K<1/2$. We then have from Theorem \ref{thm:poibin} and \eqref{eq:bimspoi} that
\begin{equation}
	C(\mathsf{Poi}_d(p_K))\leq C^{(d)}(W) \leq C(\mathsf{Poi}_d(p_1))+\exp(-d(1-Z(p_1)))-Z(p_1)^{2d}.
\end{equation}
\begin{figure*}
	\centering
	\includegraphics[width = 0.8\linewidth]{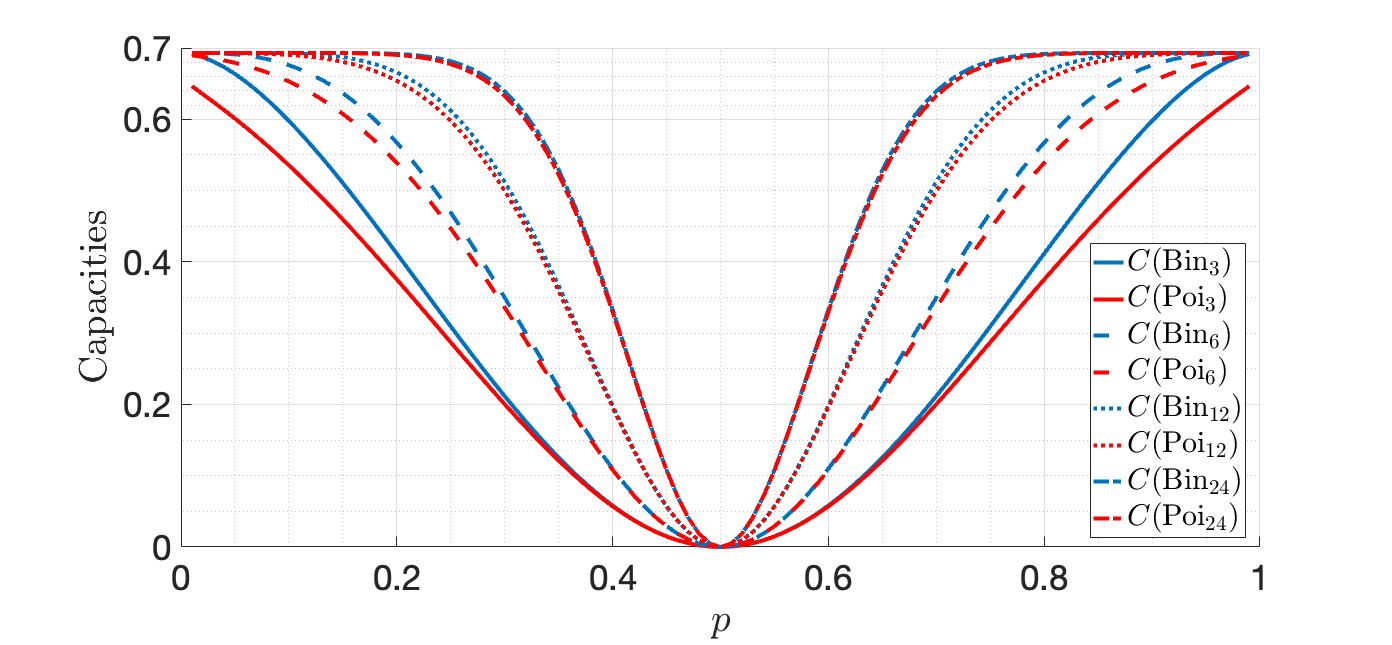}
	\caption{Plots comparing $C(\mathsf{Bin}(p))$ and $C(\mathsf{Poi}(p))$ for $p\in [0,1]$, and $d\in \{3,6,12,24\}$.}
	\label{fig:poibin}
\end{figure*}
\section{Proof Sketch of Main Results}
\label{sec:convergence}
In this section, we provide a sketch of our strategy for proving Theorem \ref{thm:conv} for a DMC $W$ with an aribitrary distribution $P_X$. To this end, we shall first consider the special case when $W$ is a BIMS channel, with a uniform input distribution, and provide a simple proof for the convergence of the mutual information $I^{(d)}$ to the input entropy. We then provide a brief sketch of the proof of Theorem \ref{thm:conv} for general DMCs, which makes use of well-known results from large-deviations theory. The full proof, owing to its length, is presented in Section \ref{sec:proofs}.
\subsection{Convergence of Mutual Information for BIMS Channels}
To build intuition, we present a simple proof for the case when the DMC $W$ is a BIMS channel, with $P_X = \text{Ber}(1/2)$. The choice of the input distribution being uniform assumes relevance since the capacity of a BIMS channel is achieved by uniform inputs (see \cite[Problem 4.8]{mct}). 

The following bounds on the conditional entropy of the input to a BIMS channel given the output, which are drawn from \cite[Prop. 2.8]{sasoglufnt}, will prove useful. Recall the definition of the Bhattacharya parameter of a BIMS channel $W$ with output alphabet $\mathcal{Y}$ given by $Z_b(W) := \sum_{y\in \mathcal{Y}} \sqrt{P_{Y|0}(y)P_{Y|1}(y)}$.
\begin{lemma}
		\label{lem:sasoglu}
	Given a BIMS channel $\overline{W}$ with input $\overline{X}\sim P_X$, where $P_X =\Ber(1/2)$, and output $\overline{Y}$, we have that 
	\begin{equation}
		Z_b(\overline{W})^2\leq \mathsf{H}(\overline{X}\mid \overline{Y})\leq \log (1 + Z_b
		(\overline{W})).
	\end{equation}
\end{lemma}
We thus obtain the following simple corollary.
\begin{corollary}
	\label{cor:sasoglu}
	Given a BIMS channel $W$ with input $X\sim P_X$, where $P_X = \Ber(1/2)$,  we have that for the $d$-view DMC $W^{(d)}$,
	\begin{equation}
	Z_b({W})^{2d}\leq \mathsf{H}({X}\mid {Y^d})\leq Z_b({W})^d.
	\end{equation}
\end{corollary}
\begin{proof}
	We first observe that when $W$ is a BIMS channel, then so is the $d$-view DMC $W^{(d)}$. To see why this is true, note that corresponding to a BIMS channel $W$, there exists a permutation $\pi: \mathcal{Y}\to \mathcal{Y}$ such that
	\begin{equation}
	W(y\mid 0) = W(\pi(y)\mid 1),
	\end{equation}
	with $\pi$ being an involution, i.e., $\pi^{-1} = \pi$. Now, for the $d$-view DMC $W^{(d)}$, we define the permutation $\pi^{(d)}: \mathcal{Y}^d\to \mathcal{Y}^d$ such that 
	\begin{equation}
	\pi^{(d)}(y_1,y_2,\ldots,y_d) = (\pi(y_1),\pi(y_2),\ldots,\pi(y_d)),
	\end{equation}
	for $y_1,\ldots,y_d\in \mathcal{Y}$. Clearly, we have that $\pi^{(d)}$ is an involution, with
	\begin{equation}
		W^{(d)}((y_1,\ldots,y_d)\mid 0) = W(\pi^{(d)}(y_1,\ldots,y_d)\mid 1),
	\end{equation}
	from \eqref{eq:d-view}. This hence implies that $W^{(d)}$ is also a BIMS channels.
	
	Now, the Bhattacharya parameter
	\begin{align}
		Z_b(W^{(d)})&= \sum_{y^d\in \mathcal{Y}^d} \sqrt{P_{Y|0}(y^d|0)P_{Y|1}(y^d|1)}\\
		&= \sum_{y^d\in \mathcal{Y}^d} \prod_{i=1}^d \sqrt{P_{Y|0}(y_i|0)P_{Y|1}(y_i|1)}\\
		&= \prod_{i=1}^d \sum_{y_i\in \mathcal{Y}} \sqrt{P_{Y|0}(y_i|0)P_{Y|1}(y_i|1)}\\
		&= Z_b(W)^d.
	\end{align}
The proof is then completed by appealing to Lemma \ref{lem:sasoglu} and noting that $\log (1+z)\leq z$, for all $z> -1$.
\end{proof}

Moreover, we have the following bound on $\rho(I)$ for general binary-input DMCs and for arbitrary input distributions $P_X$.
\begin{lemma}
	\label{lem:bimsub}
	Given a binary-input DMC $W$ with input $X\sim P_X$, where $P_X$ is arbitrary, we have that
	\begin{equation}
	\mathsf{H}(X|Y^d)\geq e^{-d\cdot \mathsf{C}(P_{Y\mid X=0},P_{Y\mid X=1})+\Theta(\log d)}.
	\end{equation}
\end{lemma}
\begin{proof}
	Observe that
	\begin{align}
		\mathsf{H}(X|Y^d)&= \E\left[\log\frac{1}{P_{X\mid Y^{d}}(X\mid Y^{d})}\right].
	\end{align}	
	Further, recall that in a binary Bayesian hypothesis testing setting (see \cite[Ch. 11]{cover_thomas}) with prior $P_{X}$
	and hypotheses $\{P_{Y\mid 0},P_{Y\mid 1}\}$, we have that the probability of correct decision is given by
	\begin{equation}
	\E\left[P(X\mid Y^{d})\right]=1-e^{-d\cdot \mathsf{C}(P_{Y\mid X=0},P_{Y\mid X=1})+\Theta(\log d)}.
	\end{equation}
	Hence, we see that for large enough $d$,
	\begin{align}
		\mathsf{H}(X|Y^d) &= \E\left[\log\frac{1}{P(X\mid Y^{d})}\right]  \\ 
		& \geq-\log\left(1-e^{-d\cdot \mathsf{C}(P_{Y\mid X=0},P_{Y\mid X=1})+\Theta(\log d)}\right) \label{eq:inter}\\
		& \geq e^{-d\cdot \mathsf{C}(P_{Y\mid X=0},P_{Y\mid X=1})+\Theta(\log d)}  ,
	\end{align}
	where the first inequality follows from Jensen's inequality and the second from the fact that $\log(1+z)\leq z$, for $z>-1$. Also, the $\Theta(\log d)$ term in the exponent in \eqref{eq:inter} is due either to the Laplace method (see \cite[Sec. 4.2]{nerifnt}) or to the results of \cite[Thm. 2.1 and Corollary 1]{arxivstat}. 
	
\end{proof}
The above ingredients suffice to afford a proof of Theorem \ref{thm:conv} for the special case of BIMS channels with a uniform input distribution.
\begin{proof}[Proof of Thm. \ref{thm:conv} for BIMS channels with uniform inputs]
	For the case when $W$ is a BIMS channel, we have that $\mathsf{C}(P_{Y|0},P_{Y|1}) = -\log Z_b(W)$ (see the remark following the statement of Theorem \ref{thm:conv}). Hence, when the input distribution is uniform, we have from Corollary \ref{cor:sasoglu}, we have that $\mathsf{H}(X|Y^d)\leq e^{-d\cdot \mathsf{C}(P_{Y|0},P_{Y|1})}$. The proof concludes by appealing then to Lemma \ref{lem:bimsub}, which shows that \begin{equation}\mathsf{H}(X|Y^d)\geq e^{-d\cdot \mathsf{C}(P_{Y|0},P_{Y|1})+\Theta(\log d)}.\end{equation} Hence, overall, we have that \begin{equation}\mathsf{H}(X|Y^d)= e^{-d\cdot \mathsf{C}(P_{Y|0},P_{Y|1})+\Theta(\log d)}.\end{equation}
\end{proof}

It is possible, with some effort, to prove Theorem \ref{thm:conv} for $I^{(d)}$ for general DMCs with arbitrary input distributions with the aid of \cite[Lemma 9]{levenshtein} (see also \cite[Thm. 5]{shgalber}) and an extension of Lemma \ref{lem:bimsub}. However, the analysis of dispersion requires more involved techniques. Owing to the length of the proofs, we  present a detailed proof for $I^{(d)}$ that introduces our techniques, in Section \ref{sec:gendmccap}, and then extend them to $V^{(d)}$, in Section \ref{sec:dispersion}. 

We now sketch a broad outline of the proof strategy for the convergence of mutual information to the input entropy. We first express the conditional entropy $\mathsf{H}(X|Y^d)$ as
\begin{align}
	\mathsf{H}(X|Y^d)&= \E\left[\log\frac{1}{P_{X\mid Y^{d}}(X\mid Y^{d})}\right]\\
	& =\sum_{x\in \mathcal{X}}P_{X}(x)\E\left[\log\frac{1}{P_{X\mid Y^{d}}(x\mid Y^{d})}\Bigg \lvert X=x\right]. 
\end{align}
We then fix an $x\in \mathcal{X}$ and focus on the inner term in the summation above. Then,
\begin{align}
	&\E\left[\log\frac{1}{P_{X\mid Y^{d}}(x\mid Y^{d})}\Bigg \lvert X=x\right]  =\int_{0}^{\infty}\P\left[\log\frac{1}{P_{X\mid Y^{d}}(x\mid Y^{d})}\geq t\Bigg \lvert X=x\right]\d t. 
\end{align}
Our approach is to obtain (exponentially tight) bounds on the probability $\P\left[-\log P_{X\mid Y^{d}}(x\mid Y^{d})\geq t \lvert X=x\right]$, using techniques from large-deviations theory---Sanov's theorem, in particular. We show that 
\begin{align}
	&\E\left[\log\frac{1}{P_{X\mid Y^{d}}(x\mid Y^{d})}\Bigg \lvert X=x\right]= \text{exp}\left({-d\cdot \min_{\tilde{x}\neq x} \mathsf{C}(P_{Y|x},P_{Y|\tilde{x}})+\Theta(\log d|\mathcal{X}|)}\right).
\end{align}	
This then allows us to use Laplace's method (e.g., see \cite[Sec. 4.2]{nerifnt}) and the independence of the input distribution $P_X$ on $d$ to conclude that
\begin{equation}\mathsf{H}(X|Y^d)= \max_{\substack{x,\tilde{x}: \tilde{x}\neq x}} \text{exp}\left({-d\cdot \mathsf{C}(P_{Y|x},P_{Y|\tilde{x}})+\Theta(\log d|\mathcal{X}|)}\right),\end{equation}
and hence that
\begin{equation}
	I^{(d)} = \mathsf{H}(X)- \textnormal{exp}\left({-d \rho +\Theta(\log d|\mathcal{X}|)}\right).
\end{equation}
Interestingly, our proof strategy allows us to obtain insights into the ``higher-order'' terms in the exponent of $\mathsf{H}(X|Y^d)$ and their dependence on both $d$ and the size of the input alphabet $|\mathcal{X}|$. The proof for the convergence of the channel dispersion to the input varentropy is similar, and relies on the results obtained during the derivation of the convergence of the mutual information to the input entropy.

In the next section, we extend Theorem \ref{thm:conv} to general \emph{multi-letter} channels, which includes as a special case, synchronization channels such as the deletion channel. We then derive explicit, analytical bounds on the rates of convergence of the mutual information and dispersion of the deletion channel to their limits.
\section{Convergence for Multi-Letter Channels and Some Examples}
\label{sec:deletion}
Recall that Theorem \ref{thm:conv} presented the convergence of mutual information and channel dispersion for $d$-view DMCs $W^{(d)}$, where the channel $W$ is specified by \emph{single-letter} transition probabilities $\{W(y|x):\ x\in \mathcal{X}, y\in \mathcal{Y}\}$. In this section, we first extend the statement of this theorem and its proof to multi-letter channels. Consider a multi-letter channel $W_{n,\nu}$, for any $n,\nu\geq 0$ that are independent of $d$, specified by a channel law $\{W(y^\nu|x^n) = P_{Y^\nu|X^n}(y^\nu|x^n):\ x^n\in \mathcal{X}^n, y^\nu\in \mathcal{Y}^\nu\}$. Here, the input alphabet is $\mathcal{X}^n$ and the output alphabet is $\mathcal{Y}^\nu$. We work with the $d$-view multi-letter channel $W_{n,\nu}^{(d)}$ defined by the channel law
\begin{equation}
P(Y^{d\nu} = y^{d\nu}\mid X^n = x^n) = \prod_{i=1}^{d} W_{n,\nu}(y^\nu|x^n).
\end{equation}
Recall here our assumption that $|\mathcal{X}|$ and $|\mathcal{Y}|$ do not depend on $d$. Similar to the earlier definition, we let 
\begin{equation}
	\label{eq:nuinf}
	I_{n,\nu}^{(d)}:= I(X^n;Y^{d\nu})
\end{equation}
to be the mutual information rate over $W_{n,\nu}^{(d)}$. Here, we assume as before that the input distribution $P_{X^n}$ is arbitrary, but fixed; in particular, it is $d$-independent. Likewise, let the channel dispersion be
\begin{equation}
	\label{eq:nudisp}
	V_{n,\nu}^{(d)}:= \mathbb{E}\left[\left(\iota(X^n;Y^{d\nu}) - I(X^n;Y^{d\nu})\right)^2\right],
\end{equation}
where $\iota(x;y^d):= \log \frac{P(y^{d\nu}|x^n)}{P(y^{d\nu})}$.
The following corollary of Theorem \ref{thm:conv} holds:

\begin{corollary}
	\label{cor:convn}
	We have that 
\begin{equation}
I_{n,\nu}^{(d)} = \mathsf{H}(X^n)- \textnormal{exp}\left({-d \rho_{n,\nu} +\Theta(\log d)+\Theta(n\cdot\log |\mathcal{X}|)}\right),
\end{equation}
and
\begin{equation}
\left \lvert V_{n,\nu}^{(d)} - \mathsf{V}(X^n)\right \rvert = \textnormal{exp}\left({-d \rho_{n,\nu} +\Theta(\log d)+ \Theta(n\cdot\log |\mathcal{X}|)}\right),
\end{equation}
where 
\begin{equation}
\rho_{n,\nu} := \min_{x^n,\tilde{x}^n: x^n\neq \tilde{x}^n} \mathsf{C}(P_{Y^\nu|x^n},P_{Y^\nu|\tilde{x}^n}). 
\end{equation}
\end{corollary}
Note that the $\Theta(n\cdot\log |\mathcal{X}|)$ terms above arise due to the fact that the input alphabet size is now $|\mathcal{X}|^n$, instead of $|\mathcal{X}|$, as in Theorem \ref{thm:conv}. Observe also, from the earlier proofs, that only the $\Theta(\log |\mathcal{X}|)$ term gets scaled by a factor of $n$ and not the $\Theta(\log d)$ term.

Now, consider the special case when $\nu = n$ and the channel $W_{n,n}$ corresponds to an $n$-letter DMC; more precisely, suppose that the multi-letter channel law obeys
\begin{equation}
W_{n,n}(y^n|x^n) = \prod_{i=1}^n W(y_i|x_i).
\end{equation}
Now, for the \emph{single-letter} channel $W$ as above, let $\rho = \min_{x,\tilde{x}: x\neq \tilde{x}} \mathsf{C}(P_{Y|x},P_{Y|\tilde{x}}) $. We then have the following corollary:
\begin{corollary}
	\label{cor:ndmc}
	For an $n$-letter DMC $W_{n,n}$,
\begin{equation}
	I_{n,n}^{(d)} = \mathsf{H}(X^n)- \textnormal{exp}\left({-dn \rho +\Theta(\log d)+n\cdot \Theta(\log |\mathcal{X}|)}\right),\ 
\end{equation}
and
\begin{equation}
	\left \lvert V_{n,n}^{(d)}- \mathsf{V}(X^n)\right \rvert = \textnormal{exp}\left({-dn \rho +\Theta(\log d)+n\cdot \Theta(\log |\mathcal{X}|)}\right).
\end{equation}

\end{corollary}
\begin{proof}
	The proof will follow from Corollary \ref{cor:convn} if we can show that $\rho_{n,n} = n\rho$, for an $n$-letter DMC $W_{n,n}$. To this end, let $u, v\in \mathcal{X}$ be such that $\rho = \mathsf{C}(P_{Y|u},P_{Y|v})$. Further,  observe that for any two inputs $x^n\neq \tilde{x}^n$, we have $\mathsf{C}(P_{Y^n|x^n},P_{Y^n|\tilde{x}^n}) = \sum_{i=1}^n \mathsf{C}(P_{Y|x_i},P_{Y|\tilde{x}_i})$. Hence,
	\begin{align}
		\rho_{n,n} &= \min_{x^n,\tilde{x}^n: x^n\neq \tilde{x}^n} \sum_{i=1}^n \mathsf{C}(P_{Y|x_i},P_{Y|\tilde{x}_i})\\
		&\geq \sum_{i=1}^n \min_{x_i,\tilde{x}_i: x_i\neq \tilde{x}_i} \mathsf{C}(P_{Y|x_i},P_{Y|\tilde{x}_i})= n\rho,
\end{align}		
where equality in the second step is obtained when $x^n = (u,u,\ldots, u)$ and $\tilde{x}^n = (v,v,\ldots,v)$.
\end{proof}

While Corollary \ref{cor:convn} treats $n$ and $\nu$ as fixed constants, one can generalize the definitions of the $d$-view mutual information and channel dispersion to the setting where $\nu$ is a \emph{random variable} too. In particular, let $I_n^{(d)}$ and $V_n^{(d)}$ be defined the same way as in \eqref{eq:nuinf} and \eqref{eq:nudisp}, respectively, but where the expectations are taken with respect to the randomness in the length $\nu$ of the outputs too. Furthermore, let
\begin{equation}
f_\lambda(x^n,\tilde{x}^n):= \sum_{m\geq 0}\Pr[\nu = m]\sum_{y^m}P^{1-\lambda}(y^m|x^n)P^\lambda(y^m|\tilde{x}^n).
\end{equation}
We then define 
\begin{equation}
\rho_{n} := \min_{ x^n\neq \tilde{x}^n} \max_{\lambda\in [0,1]} -\log f_\lambda(x^n,\tilde{x}^n).
\end{equation}


In what follows, we present some estimates of $\rho_{n}$ for the deletion channel, where the length $\nu\leq n$ of the outputs is random. 

\subsection{Upper Bounds on $\rho_n$ for the Deletion Channel}

In this subsection, we consider the multi-letter deletion channel, which is an important example of a synchronization channel \cite{mahdi_survey}. Recall that for $\delta\in (0,1)$, the deletion channel $\text{Del}(\delta)$ deletes each input bit $x_i\in \mathcal{X}$, $1\leq i\leq n$, with probability $\delta$. 


For Del$(\delta)$, for any fixed $x^n,\tilde{x}^n$ with $x^n\neq \tilde{x}^n$, and for any fixed $\lambda\in (0,1)$, we have
\begin{align}
	&f_\lambda(x^n,\tilde{x}^n) = \sum_{m=0}^{n}\delta^{n-m}(1-\delta)^{m}\sum_{y^m}N^{1-\lambda}(x^{n}\to y^{m})\cdot N^{\lambda}(\tilde{x}^{n}\to y^{m}), \label{eq:del1}
\end{align}
where $N(u^n\to v^m)$, for $u^n\in \{0,1\}^n$, $v^m\in \{0,1\}^m$, is the number of occurrences of $v^m$ as a (potentially non-contiguous) subsequence in $u^n$.

 In what follows, we present three upper bounds on $\rho_{n}$ for Del$(\delta)$, and via these upper bounds, arrive at some interesting properties of $\rho_n$ and comparisons with the rates $\rho_{n,n}$ of the BEC$(\delta)$ and BSC$(\delta)$. The upper bounds are obtained by fixing specific sequences $x^n$, $\tilde{x}^n \in \mathcal{X}^n$ and using $\max_{\lambda\in [0,1]} -\log f_\lambda(x^n,\tilde{x}^n)$, evaluated at these chosen sequences, as an upper bound on $\rho_n$. The first of our upper bounds is na\"ive and will be subsumed by the later bounds.
 
 \begin{lemma}
 	\label{lem:del1}
 	For $\Del(\delta)$, we have that
 	$\rho_{n}\leq -n\log {\delta}$.
 \end{lemma}
\begin{proof}
	Consider \emph{any} fixed $x^n,\tilde{x}^n$ with $x^n\neq \tilde{x}^n$. For any $\lambda\in [0,1]$, $f_\lambda(x^n,\tilde{x}^n)\geq \delta^n$, from \eqref{eq:del1}, since the empty string occurs exactly once as a subsequence of both $x^n$ and $\tilde{x}^n$. Hence, we have that 
 \begin{equation}
 \rho_{n} = \min_{ x^n\neq \tilde{x}^n} \max_{\lambda\in [0,1]} -\log f_\lambda(x^n,\tilde{x}^n) \leq 	-n\log {\delta}.
 \end{equation}
\end{proof}
As expected, from the above lemma and from Corollary \ref{cor:ndmc}, we see that the rate $\rho_n$ for the deletion channel is \emph{smaller} than the rate $\rho_{n,n}$ for the $n$-letter erasure channel with erasure probability $\delta$. Indeed, recall that for the BEC$(\delta)$, we have that $\rho = -\log \delta$, and hence, by Corollary \ref{cor:ndmc}, $\rho_{n,n} = -n\log\delta$. Note also that by an application of the data processing inequality, it follows that the capacity of $\Del(\delta)$ is at most the capacity of the BEC$(\delta)$, for $\delta\in [0,1]$ \cite[Sec. 5.2]{mahdi_survey}.

The bound that follows improves on this na\"ive upper bound, for a certain range of $\delta$ values, and for large enough $n$. For simplicity, in the second part of the next lemma, we assume that $n$ is even; the case of odd $n$ can be handled with minor modifications to the exisiting proof.

 \begin{lemma}
 	\label{lem:del2}
 	For $\Del(\delta)$, we have that $\limsup_{n\to \infty} \rho_n = 0$, if $\delta>1/2$.
 	
%
 \end{lemma}
\begin{proof}
	Fix $\delta >1/2$. We pick $x^n, \tilde{x}^n$ to be, respectively, the alternating sequences $01010\ldots$ and $10101\ldots$, each having length $n$, for the purposes of this lemma. 
	
	Now, for any $y^m\in \{0,1\}^m$, with $m\leq n$, we see that $N(x^{n}\to y^{m}) = N(\tilde{x}^{n}\to y^{m})$. Furthermore, $N(x^{n}\to y^{m})>0$ iff $w(y^m)\in [\max\{0,m-n/2\},n/2]$; we then call $y^m$ ``admissible''. For such admissible $y^m$, we have that
	\begin{equation}
	N(x^{n}\to y^{m}) = \binom{n/2}{w(y^m)}\cdot \binom{n/2}{m-w(y^m)} = N(\tilde{x}^{n}\to y^{m}).
	\end{equation}
	The above expression is obtained by choosing $w(y^m)$ locations among the $n/2$ $1$s in $x^n$ (resp. $\tilde{x}^n$) as the locations of the $1$s in $y^m$, and likewise, $n-w(y^m)$ locations among the $n/2$ $0$s in $x^n$ (resp. $\tilde{x}^n$) as the locations of the $0$s in $y^m$.
	Furthermore, in the above expression, we implicitly assume that the binomial coefficient $\binom{k}{t} = 0$ if $t<0$ or $t>k$. 
	
	Thus, for $\lambda\in (0,1)$, 
	\begin{align}
		f_\lambda(x^n,\tilde{x}^n)
		&= \sum_{m=0}^{n}\delta^{n-m}(1-\delta)^{m}\sum_{y^m \text{ admissible}} \binom{n/2}{w(y^m)} \binom{n/2}{m-w(y^m)}  \\
		&= \sum_{m=0}^{n}\delta^{n-m}(1-\delta)^{m}\sum_{w=\max\{m-n/2,0\}}^{\min\{m,n/2\}} \binom{m}{w} \binom{n/2}{w} \binom{n/2}{m-w}. \label{eq:del2}
	\end{align}
Let us call the inner summation above as $\beta(m)$. Now, note that
\begin{equation}
\beta(m) = 
\begin{cases}
	\sum_{w=0}^m \binom{m}{w} \binom{n/2}{w}\binom{n/2}{m-w},\ \text{if } m\leq n/2,\\
	\sum_{w=m-n/2}^{n/2} \binom{m}{w} \binom{n/2}{w}\binom{n/2}{m-w},\ \text{if } m> n/2.
\end{cases}
\end{equation}
After some algebra, we can write
\begin{equation}
	\beta(m) = 
	\begin{cases}
		\sum_{w=0}^m \binom{m}{w} \binom{n/2}{w}\binom{n/2}{m-w},\ \text{if } m\leq n/2,\\
		\sum_{w=0}^{n-m} \binom{m}{n/2-w} \binom{n/2}{w}\binom{n/2}{m-(n/2-w)},\  \text{if } m> n/2.
	\end{cases}
\end{equation}

%
Let \begin{equation}\alpha_1(m):= \sum_{w=0}^{m} \binom{m}{w} \binom{n/2}{w}\binom{n/2}{m-w}\end{equation} and \begin{equation}\alpha_2(m):= \sum_{w=0}^{n-m} \binom{m}{n/2-w} \binom{n/2}{w} \binom{n/2}{m-(n/2-w)}.\end{equation} Consider $\alpha_1(m)$, for the case when $m\leq n/2$. We then have that
\begin{align}
\alpha_1(m)&\geq \min_{0\leq \omega\leq m} \binom{n/2}{m-\omega}\cdot \sum_{w=0}^{m} \binom{m}{w} \binom{n/2}{w}\\
 &= \binom{n/2}{0}\binom{n/2+m}{m} = \binom{n/2+m}{m},
\end{align}	
where the first equality follows from the Chu-Vandermonde identity (see, e.g., \cite[p. 42]{chu}).
Likewise, consider $\alpha_2(m)$, for the case when $m\geq n/2$. Note that
\begin{align}
	\alpha_2(m)&=\sum_{w=0}^{n-m} \binom{m}{n/2-w} \binom{n/2}{w} \binom{n/2}{m-(n/2-w)}\\
	&= \sum_{w=0}^{n-m} \binom{m}{n/2-w} \binom{n/2}{w} \binom{n/2}{n-m-w}\\
	&\geq \sum_{w=0}^{n-m} \binom{n/2}{w} \binom{n/2}{n-m-w}= \binom{n}{m},
\end{align}	
where the last equality above again uses the Chu-Vandermonde identity. Thus, we obtain that  
\begin{equation}
\beta(m) \geq
\begin{cases}
	\binom{n/2+m}{m},\ \text{if } m\leq n/2,\\
	\binom{n}{m},\  \text{if } m> n/2.
\end{cases}
\end{equation}
Plugging this back into \eqref{eq:del2}, we see that
\begin{align}
f_\lambda(x^n,\tilde{x}^n)&\geq \sum_{m=0}^{n/2}\delta^{n-m}(1-\delta)^{m}\cdot \binom{n/2+m}{m} + \sum_{m=n/2+1}^{n}\delta^{n-m}(1-\delta)^{m}\cdot \binom{n}{m}.
\label{eq:del3}
\end{align}	
Call the first summation above as $\gamma_1$ and the second as $\gamma_2$. Now, observe that $\gamma_2$ is precisely the probability $\Pr[L>n/2]$, for a binomial random variable $L\sim \text{Bin}(n,\delta)$. It is well-known (see, e.g., \cite[Example 1.6.2]{gallager_stoch}) that for $\delta>1/2$, this probability is at least \begin{equation}1-\text{exp}(-n\cdot D(\text{Ber}(1/2)||\text{Ber}(\delta))) = 1-Z(\delta)^n,\end{equation} where $Z(\delta) = 2\sqrt{\delta(1-\delta)}$ is the Bhattacharya parameter of the BSC$(\delta)$. 
Furthermore, note that
\begin{align}
	\gamma_1&\geq (\delta(1-\delta))^{n/2}\cdot \binom{n}{n/2}\\
	&=  (\delta(1-\delta))^{n/2}\cdot 2^{n-O(\log n)}\\ &= Z(\delta)^n\cdot 2^{-O(\log n)}.
\end{align}	
where the first equality follows from \cite[Eq. (5.38)]{asymp}.

Hence,  we obtain that for $\lambda\in (0,1)$,
\begin{align}
	f_\lambda(x^n,\tilde{x}^n)\geq 
		1-Z(\delta)^n\cdot(1-2^{-O(\log n)}).
\end{align}	

Using the previous inequality, we obtain that when $\delta>1/2$, 
\begin{align}
	\limsup_{n\to \infty} \rho_n 
	&\leq \limsup_{n\to \infty} \max_{\lambda\in [0,1]} -\log f_\lambda(x^n,\tilde{x}^n)\\
	&\leq \limsup_n \max\{0, -\log (1-Z(\delta)^n\cdot(1-2^{-O(\log n)}))\}\\
	&= 0,
\end{align}	
where the last inequality uses the fact that for $\lambda\in \{0,1\}$, we have $f_\lambda(x^n,\tilde{x}^n) = 1$. 
\end{proof}



Next, we show that for all $\delta\in (0,1)$ (and hence for $\delta\leq 1/2$ too), the rate of convergence $\rho_n$ can in fact be bounded from above by a constant, for large enough $n$. Before we do so, we state and prove a useful lemma.

\begin{lemma}
	\label{lem:del3}
	For $\Del(\delta)$, where $\delta\in (0,1)$, we have that
	\begin{equation}
		\rho_{n}\leq \max_{\lambda\in [0,1]} \lambda\log n - \log \E_{L\sim \text{Bin}(n,\delta)}\left[L^\lambda\right].
\end{equation}		
\end{lemma}
\begin{proof}
	As before, we pick specific sequences $x^n, \tilde{x}^n$ and seek to upper bound $\rho_{n,\nu}$ by $\max_{\lambda\in [0,1]}-\log f_\lambda(x^n,\tilde{x}^n)$, for this choice of $x^n, \tilde{x}^n$. Here, we choose $x^n = 0^n$ and $\tilde{x}^n = 0^{n-1}1$. 
	
	Clearly, for any $m\leq n$, we have that $y^m\in \{0,1\}^n$ is a subsequence of $x^n$ iff $y^m = 0^m$. For $y^m = 0^m$, hence, it follows that
	\begin{equation}
	N(x^n\to y^m) = \binom{n}{m}\ \text{and}\ N(\tilde{x}^n\to y^m) = \binom{n-1}{m}.
	\end{equation}
	Thus, for $\lambda\in (0,1)$, 
		\begin{align}
		f_\lambda(x^n,\tilde{x}^n) 
		&= \sum_{m=0}^{n}\delta^{n-m}(1-\delta)^{m}\cdot \binom{n}{m}^{1-\lambda}\cdot \binom{n-1}{m}^\lambda  \\
		&= \sum_{m=0}^{n}\delta^{n-m}(1-\delta)^{m}\cdot\binom{n}{m}\cdot \left(\frac{n-m}{n}\right)^\lambda  \\
		&= \mathbb{E}_{L\sim \text{Bin}(n,1-\delta)}\left[\left(1-\frac{L}{n}\right)^\lambda\right]  \\
		&= \frac{1}{n^\lambda}\cdot \mathbb{E}_{L\sim \text{Bin}(n,\delta)}\left[{L}^\lambda\right],
	\end{align}
and the claim follows by noting that $f_\lambda(x^n,\tilde{x}^n) = 1$, for $\lambda\in \{0,1\}$.
\end{proof}

We next use Lemma \ref{lem:del3} to show that for large enough $n$, $\rho_n$ can be bounded from above by a \emph{constant}. We thus obtain the following corollary:
\begin{corollary}
	\label{cor:delta}
	For $\Del(\delta)$, where $\delta\in (0,1)$, we have that $\limsup_{n\to \infty} \rho_{n}\leq -\log \delta$.
\end{corollary}
\begin{proof}
	The proof proceeds by obtaining estimates of the fractional moments of a binomial random variable used in Lemma \ref{lem:del3}. Let $\beta, \eta>0$ be small real numbers. We then have that for $L\sim \text{Bin}(n,\delta)$, and $\lambda\in (0,1)$,
	\begin{align}
		\E[L^\lambda]&\geq (n(\delta-\beta))^\lambda\cdot \Pr[L\geq n(\delta-\beta)] \\
		&\geq (n(\delta-\beta))^\lambda\cdot (1-\text{exp}(-n\cdot D(\text{Ber}(\delta-\beta)||\text{Ber}(\delta)))) \\
		&\geq (1-\eta)\cdot(n(\delta-\beta))^\lambda, \label{eq:del4}
	\end{align}
	for $n$ large enough.
	
	Plugging \eqref{eq:del4} into Lemma \ref{lem:del3}, we obtain that for $n$ large enough,
	\begin{align}
		\rho_n &\leq \max_{\lambda\in [0,1]} \lambda\log n - \log \E_{L\sim \text{Bin}(n,\delta)}\left[L^\lambda\right]\\
		&\leq \max\bigg\{0, \max_{\lambda\in (0,1)} \lambda\log n-\lambda\log(n(\delta-\beta))-\log(1-\eta)\bigg\}\\
		&= \max\bigg\{0,\max_{\lambda\in (0,1)} -\lambda \log(\delta-\beta)-\log (1-\eta)\bigg\}\\
		&\leq -(1-\zeta)\log \delta,
\end{align}	
for small $\zeta>0$. Thus, taking $\limsup_{n\to \infty}$ on both sides of the inequality above and letting $\zeta\downarrow 0$, we obtain the statement of the corollary.
%
\end{proof}
Hence, putting together Lemma \ref{lem:del2} and Corollary \ref{cor:delta}, we see that for the deletion channel Del$(\delta)$, the rate $\rho_n$, for large enough $n$, is bounded above by a fixed constant, for $\delta\leq 1/2$, and by an arbitrarily small constant, for $\delta>1/2$.

In the next section, we take a detour from the asymptotic regime of large $d$ considered in the rest of the paper, and  seek to obtain tight, non-asymptotic bounds on the capacity of a special multi-view DMC, which holds for all values of $d$.
\section{Poisson Approximation Channel}
\label{sec:poisson}
In this section, we consider the special case of the $d$-view BSC$(p)$, denoted by $\text{BSC}^{(d)}(p)$. Recall that the capacity of this channel equals the capacity of the binomial channel $\mathsf{Bin}_d(p)$, which is given in \eqref{eq:bin}. Our objective is to show that the Poisson approximation channel $\mathsf{Poi}_d(p)$ introduced in Section  \ref{sec:main} has capacity that is a very good lower bound on the capacity of $\text{BSC}^{(d)}(p)$, via the proof of Theorem \ref{thm:poibin}. Before we proceed with the proof, we state and prove a simple lemma about the concavity of the mutual information $I^{(d)} = I(X;Y^d)$, for any $d$-view DMC $W^{(d)}$ and input distribution $P_X$.

\begin{lemma}
	\label{lem:concave}
	For any $d$-view DMC $W^{(d)}$ and input distribution $P_X$, we have that $I^{(d)}$ is concave over the positive integers, i.e.,
	\begin{equation}
	I^{(d)} - I^{(d-1)}\geq I^{(d+1)} - I^{(d)},
	\end{equation}
	for all $d\geq 1$.
\end{lemma}
\begin{proof}
	We first rewrite what we intend to prove differently: since $I^{(m)} = \mathsf{H}(X) + \mathsf{H}(Y^m) - \mathsf{H}(X,Y^m)$, for any $m\geq 0$, with $I^{(0)}:= 0$, we would like to show that
	\begin{align}
	\mathsf{H}(Y^d) - &\mathsf{H}(Y^{d-1})+\mathsf{H}(X,Y^{d-1})-\mathsf{H}(X,Y^d) \nonumber  \\
	&\geq \mathsf{H}(Y^{d+1}) - \mathsf{H}(Y^{d})+\mathsf{H}(X,Y^{d})-\mathsf{H}(X,Y^{d+1}). \label{eq:toshowconcave}
	\end{align}
However, observe that for any $m\geq 0$, $\mathsf{H}(X,Y^{m+1})-\mathsf{H}(X,Y^{m}) = \mathsf{H}(Y_{m+1}|X,Y^m) = \mathsf{H}(Y_{m+1}|X)$, by the channel law of $W^{(d)}$. Hence, \eqref{eq:toshowconcave} reduces to showing that
\begin{equation}
\mathsf{H}(Y^d) - \mathsf{H}(Y^{d-1})\geq \mathsf{H}(Y^{d+1}) - \mathsf{H}(Y^{d}). \label{eq:toshowconcave2}
\end{equation}
Now, by the submodularity of entropy \cite[Thm. 1.6]{polyanskiywu}, we have that 
\begin{equation}
\mathsf{H}(Y^{d-1},Y_{d+1}) - \mathsf{H}(Y^{d-1})\geq \mathsf{H}(Y^{d+1}) - \mathsf{H}(Y^{d}).
\end{equation}
The above implies \eqref{eq:toshowconcave2}, since the distributions of $(Y^{d-1},Y_{d+1})$ and $Y^d$ are equal, by the stationarity of the DMC $W^{(d)}$.
\end{proof}

For clarity, we first write down the capacity $C(\mathsf{Poi}_d(p))$ of the Poisson approximation channel. As argued in Section \ref{sec:main}, the capacity is achieved when the inputs $X\sim P_X = \text{Ber}(1/2)$, yielding
\begin{align}
	C(\mathsf{Poi}_d(p)) 
    &= \mathsf{H}(X) - \mathsf{H}(X|R_1,R_2)  \\
	&= 1+\sum_{\substack{r_1\geq 0,\\ r_2\geq 0}} \frac{e^{-dp}(dp)^{r_1}}{r_1!}\cdot \frac{e^{-d(1-p)}(d(1-p))^{r_2}}{r_2!}\cdot \log \frac{p^{r_1}(1-p)^{r_2}}{p^{r_1}(1-p)^{r_2}+p^{r_2}(1-p)^{r_1}}. \label{eq:poicap}
\end{align}
We now prove Theorem \ref{thm:poibin}.
\begin{proof}[Proof of Thm. 3]
	Consider the mutual information between the input and outputs of the $\mathsf{Poi}_d(p)$ channel when $P_X$ is uniform over the inputs. Since the number of outputs $N=R_1+R_2$ is known to the decoder, we have that
	\begin{align}
		C(\mathsf{Poi}_d(p))&= I(X;R_1,R_2)  \\
		&= I(X;R_1,R_2\mid N)   \\
		 &= \sum_{n=0}^{\infty}\P[N=n]\cdot I(X;R_1,R_2\mid N=n). \label{eq:poiexp}
	\end{align}
Now, observe that since the sum of a Poi$(dp)$ random variable and a Poi$(d(1-p))$ random variable is distributed according to Poi$(d)$, we have by symmetry that $N\sim \text{Poi}(d)$. Further, it is easy to check that the conditional distribution $P_{R_1,R_2\mid N=n}$ obeys $P_{R_1,R_2\mid N=n}(r_1,r_2\mid n) = \binom{n}{r_1}p^{r_1}(1-p)^{n-r_1}$, for all $r_1,r_2\geq 0$ such that $r_1+r_2 = n$. Hence, we have that $I(X;R_1,R_2\mid N=n) = C(\mathsf{Bin}_n(p))$.

Thus, substituting in \eqref{eq:poiexp}, we get that
\begin{equation}
	\label{eq:expectation}
C(\mathsf{Poi}_d(p)) = \E_{N\sim \text{Poi}(d)}\left[C(\mathsf{Bin}_N(p))\right].
\end{equation}
 Next, using the above pleasing expression, we show that $C(\mathsf{Poi}_d(p))$ is indeed a lower bound on $C(\mathsf{Bin}_d(p))$. From Lemma \ref{lem:concave} we have that $C(\mathsf{Bin}_n(p))$ is concave over the positive integers. Therefore, applying Jensen's inequality, we get that
 \begin{align}
 C(\mathsf{Poi}_d(p)) &= \E_{N\sim \text{Poi}(d)}\left[C(\mathsf{Bin}_N(p))\right]\\
 &\leq C(\mathsf{Bin}_{\E N}(p)) = C(\mathsf{Bin}_{d}(p)).
\end{align}	
Furthermore, we have that
\begin{align}
	C(\mathsf{Bin}_d(p)) - C(\mathsf{Poi}_d(p))&= C(\mathsf{Bin}_d(p)) - \E_{N\sim \text{Poi}(d)}\left[C(\mathsf{Bin}_N(p))\right]\\
	&\leq C(\mathsf{Bin}_d(p))-1+\E_{N\sim \text{Poi}(d)}[Z(p)^N]\\
	&\leq \text{exp}(-d(1-Z(p)))-Z(p)^{2d},
\end{align}	
where the two inequalities make use of Corollary \ref{cor:sasoglu} and the last inequality uses the well-known expression for the moment-generating function of a Poisson random variable.
\end{proof}
\begin{remark}
	For the binomial channel $\overline{\mathsf{Bin}}_n$, whose input is a real number in $[0,1]$, with $n$ output trials, that is considered in \cite{binomialcomposite} too, a suitable Poisson approximation channel $\overline{\mathsf{Poi}}_n$ can be designed. This channel takes as input $X\in [0,1]$ and outputs the pair $(R_1,R_2)$ such that $P_{R_1|X} = \text{Poi}(nX)$ and $P_{R_2|X} = \text{Poi}(n(1-X))$. For $\overline{\mathsf{Poi}}_n$, as above, we will have that  $C(\overline{\mathsf{Poi}}_n) = \E_{N\sim \text{Poi}(n)}\left[C(\overline{\mathsf{Bin}}_N)\right]$. However, the proof of the concavity of the capacity $C(\overline{\mathsf{Bin}}_n)$ as in Lemma \ref{lem:concave} will not hold, owing to the failure of the conditional independence assumption that is crucially employed in the proof. We are hence unable to obtain a straightforward lower bound on $C(\overline{\mathsf{Bin}}_n)$ via $C(\overline{\mathsf{Poi}}_n)$, using the previous techniques.
\end{remark}
\section{Proofs}
\label{sec:proofs}
In this section, we provide a detailed proof of our main result, Theorem \ref{thm:conv}.
\subsection{Convergence of Mutual Information for General DMCs}
\label{sec:gendmccap}
Given a general DMC $W$ with input alphabet $\mathcal{X}$ (not necessarily binary) and output alphabet $\mathcal{Y}$, consider the  $d$-view DMC $W^{(d)}$. Let the input to $W^{(d)}$ be $X$, with input distribution being $P_X$, and let its outputs be $Y_1,\ldots,Y_d \in \mathcal{Y}^d$. The joint distribution of the inputs and outputs is such that $P_{X,Y^d}(x,y^d) = P_X(x)\prod_{i=1}^d W(y_i|x)$. Much like the proof of Corollary \ref{cor:sasoglu}, the inequality
\begin{equation}
	\mathsf{H}(X|Y^d) \leq Z_g(W)^d,
\end{equation}
can be shown to hold (see \cite[Prop. 4.8]{sasoglufnt}), where \begin{equation}Z_g(W) := \sum_{x\neq x'} \sum_{y\in \mathcal{Y}} \sqrt{P(x)W(y|x)P(x')W(y|x')}\end{equation} is a scaled version of the Bhattacharya parameter for general DMCs. This then results in $\rho(I)\geq -\log Z_g(W)$, thereby showing that for well-behaved DMCs, the rate of convergence of $I^{(d)}$ to $\mathsf{H}(X)$ is exponential in $d$. We shall now take up the question of the exact speed of this convergence.

To this end, we will be next interested in evaluating, up to exponential tightness, the behaviour of
\begin{align}
	\mathsf{H}(X|Y^d)&= \E\left[\log\frac{1}{P_{X\mid Y^{d}}(X\mid Y^{d})}\right]\\
	& =\sum_{x\in \mathcal{X}}P_{X}(x)\E\left[\log\frac{1}{P_{X\mid Y^{d}}(x\mid Y^{d})}\Bigg \lvert X=x\right]. \label{eq:inter2}
\end{align}
Now, fix an $x\in \mathcal{X}$ and focus on the inner term in \eqref{eq:inter2}. Then,
\begin{align}
	\E\left[\log\frac{1}{P_{X\mid Y^{d}}(x\mid Y^{d})}\Bigg \lvert X=x\right] =\int_{0}^{\infty}\P\left[\log\frac{1}{P_{X\mid Y^{d}}(x\mid Y^{d})}\geq t\Bigg \lvert X=x\right]\d t. \label{eq:inter3}
\end{align}
The inner probability, to wit, \begin{equation}p_x(t):=\P\left[-\log P_{X\mid Y^{d}}(x\mid Y^{d})\geq t \lvert X=x\right]\end{equation} is bounded as follows:
\begin{lemma}
	\label{lem:inter}
	We have that
	\begin{equation}
		\frac{p_x(t)}{(|{\cal X}|-1)}\leq\max_{\tilde{x}\neq x}\P\left[\frac{P(x)P(Y^{d}\mid x)}{P(\tilde{x})P(Y^{d}\mid\tilde{x})}\leq \frac{e^{-(t-\log(|{\cal X}|-1))}}{1-e^{-t}}\Bigg \lvert X=x\right]
  \end{equation}
  and
  \begin{equation}
	p_x(t)	\geq \max_{\tilde{x}\neq x}\P\left[\frac{P(x)P(Y^{d}\mid x)}{P(\tilde{x})P(Y^{d}\mid\tilde{x})}\leq\frac{e^{-t}}{1-e^{-t}}\Bigg \lvert X=x\right].
	\end{equation}	
\end{lemma}
\begin{proof}
	First, observe that
	\begin{align}
		 p_x(t) 
		& =\P\left[\frac{P_{XY^{d}}(x,Y^{d})}{P_{Y^{d}}(Y^{d})}\leq e^{-t}\Bigg \lvert X=x\right]\\
		& =\P\left[P(x)P(Y^{d}\mid x)\leq e^{-t}\sum_{\tilde{x}}P(\tilde{x})P(Y^{d}\mid\tilde{x})\Big \lvert X=x\right]\\
		& =\P\left[\frac{P(x)P(Y^{d}\mid x)}{\sum_{\tilde{x}\neq x}P(\tilde{x})P(Y^{d}\mid\tilde{x})}\leq\frac{e^{-t}}{1-e^{-t}}\Big \lvert X=x\right].
	\end{align}	
	We first prove the upper bound in the lemma. The above probability obeys
	\begin{align}
		\P\left[\frac{P(x)P(Y^{d}\mid x)}{\sum_{\tilde{x}\neq x}P(\tilde{x})P(Y^{d}\mid\tilde{x})}\leq\frac{e^{-t}}{1-e^{-t}}\Bigg \lvert X=x\right]
		& \leq\P\left[\frac{P(x)P(Y^{d}\mid x)}{(|{\cal X}|-1)\max_{\tilde{x}\neq x}P(\tilde{x})P(Y^{d}\mid\tilde{x})}\leq\frac{e^{-t}}{1-e^{-t}}\Bigg \lvert X=x\right]\\
		& =\P\left[\frac{1}{|{\cal X}|-1}\cdot\min_{\tilde{x}\neq x}\frac{P(x)P(Y^{d}\mid x)}{P(\tilde{x})P(Y^{d}\mid\tilde{x})}\leq\frac{e^{-t}}{1-e^{-t}}\Bigg \lvert X=x\right]\\
		& \leq\sum_{\tilde{x}\neq x}\P\left[\frac{1}{|{\cal X}|-1}\cdot\frac{P(x)P(Y^{d}\mid x)}{P(\tilde{x})P(Y^{d}\mid\tilde{x})}\leq\frac{e^{-t}}{1-e^{-t}}\Bigg\lvert X=x\right]\\
		& \leq(|{\cal X}|-1)\cdot\max_{\tilde{x}\neq x}\P\left[\frac{1}{|{\cal X}|-1}\frac{P(x)P(Y^{d}\mid x)}{P(\tilde{x})P(Y^{d}\mid\tilde{x})}\leq\frac{e^{-t}}{1-e^{-t}}\Bigg \lvert X=x\right],
	\end{align}	
	thereby showing the upper bound. Further, for any $\tilde{x}$, we have that
	\begin{align}
		p_x(t) &=
		\P\left[\frac{P(x)P(Y^{d}\mid x)}{\sum_{\tilde{x}\neq x}P(\tilde{x})P(Y^{d}\mid\tilde{x})}\leq\frac{e^{-t}}{1-e^{-t}}\mid X=x\right]\\
		&\geq\max_{\tilde{x}\neq x}\P\left[\frac{P(x)P(Y^{d}\mid x)}{P(\tilde{x})P(Y^{d}\mid\tilde{x})}\leq\frac{e^{-t}}{1-e^{-t}}\mid X=x\right],
	\end{align}	
	thereby proving the lower bound.
\end{proof}
Since we are interested mainly in the behaviour of $\mathsf{H}(X|Y^d)$ up to first order in the exponent $d$, it suffices for us to understand the behaviour of $p_x(t)$ upto first order in the exponent $d$ with the expectation that we will accrue an additional $\Theta(\log d)$ term in the exponent, in the final expression (see Eq. \eqref{eq:inter3} above and  \cite[Sec. 4.2]{nerifnt} for intuition). To this end, following Lemma \ref{lem:inter}, for fixed $z\in [0,\infty)$ and $c>0$, we shall focus on the term
\begin{equation}
	\Gamma_{x,\tilde{x}}(z):= \P\left[\frac{P(x)P(Y^{d}\mid x)}{P(\tilde{x})P(Y^{d}\mid\tilde{x})}\leq \frac{ce^{-z}}{1-e^{-z}}\Bigg \lvert X=x\right],
\end{equation}
with the reasoning that the rate of convergence of $\mathsf{H}(X|Y^d)$ will be dictated by the smallest rate of $\Gamma_{x,\tilde{x}}(z)$, as $z$ increases from $0$ to $\infty$.

Before we proceed, we require some more notation. Let us denote the log-likelihood ratio by 
\begin{equation}
	L_{x,\tilde{x}}^y:=\log\frac{P_{Y\mid X}(y\mid x)}{P_{Y\mid X}(y\mid\tilde{x})}.
\end{equation}
For a given $y^{d}\in{\cal Y}^{d}$ and for $b\in \mathcal{Y}$, let
$\hat{Q}_{y^{d}}(b):=\frac{1}{d}\sum_{i=1}^{d}\mathds{1}\{y_{i}=b\}$ be the empirical frequency of the letter $b$ in any output sequence. Further, we define
\begin{equation}
	L_{x,\tilde{x}}(\hat{Q}):=\sum_{b\in{\cal Y}}\hat{Q}(b)L_{x,\tilde{x}}^b.
\end{equation}
Finally, let 
\begin{equation}
	v_{x,\tilde{x}}\equiv v_{x,\tilde{x}}(z) := \frac{1}{d}\left(\log\left(\frac{ce^{-z}}{1-e^{-z}}\right) - \log \frac{P(x)}{P(\tilde{x})}\right)
\end{equation}
and let $z(v_{x,\tilde{x}}):= \log\left(1+\frac{cP(\tilde{x})}{P(x)}\cdot e^{-dv_{x,\tilde{x}}}\right)$ be its inverse function.
\begin{proposition}
	\label{prop:Gamma}
	We have that for any fixed $\delta>0$,
	\begin{enumerate}
		\item If $v_{x,\tilde{x}}\geq L_{x,\tilde{x}}({P_{Y|x}}) + \delta$, then $\Gamma_{x,\tilde{x}}(z(v_{x,\tilde{x}})) = \Theta(1)$.
		\item If $v_{x,\tilde{x}}\leq L_{x,\tilde{x}}({P_{Y|x}})-\delta$, then $\Gamma_{x,\tilde{x}}(z(v_{x,\tilde{x}})) = e^{-dE(v_{x,\tilde{x}})+\Theta(\log d)}$, where
		\begin{align}
			E(v_{x,\tilde{x}}):= 
			&\begin{cases}
				\max\limits_{\lambda\geq 0}\bigg[-\log \left(\sum_{y} P(y|x)^{1-\lambda}P(y|x')^\lambda\right)-\lambda v_{x,\tilde{x}}\bigg],\ \text{if }v_{x,\tilde{x}}\geq \min_{b\in \mathcal{Y}} L_{x,\tilde{x}}^b,\\
				\infty,\ \text{otherwise}.
			\end{cases}
		\end{align}
	\end{enumerate}
\end{proposition}
\begin{proof}
	Observe that for any $z\in [0,\infty)$,
	\begin{align}
		\Gamma_{x,\tilde{x}}(z) &= \P\left[\frac{P(x)P(Y^{d}\mid x)}{P(\tilde{x})P(Y^{d}\mid\tilde{x})}\leq \frac{ce^{-z}}{1-e^{-z}}\Bigg \lvert X=x\right] \\
		& =\P\left[\sum_{i=1}^{d}\log\frac{P_{Y\mid X}(Y_{i}\mid x)}{P_{Y\mid X}(Y_{i}\mid\tilde{x})}\leq\log\left(\frac{ce^{-z}}{1-e^{-z}}\right) - \log \frac{P(x)}{P(\tilde{x})}\Bigg \lvert X=x\right]\\
		& =\P\left[\sum_{i=1}^{d}L_{x,\tilde{x}}^{Y_{i}}\leq\log\left(\frac{ce^{-z}}{1-e^{-z}}\right)- \log \frac{P(x)}{P(\tilde{x})}\Bigg \lvert X=x\right]\\
		& =\P\left[\E_{\hat{Q}}\left[L_{x,\tilde{x}}^{Y^{d}}\right]\leq v_{x,\tilde{x}}(z)\Big \lvert X=x\right].
	\end{align}	
	Writing this differently, we have that for any $v_{x,\tilde{x}}\in \mathbb{R}$,
	\begin{equation}
		\Gamma_{x,\tilde{x}}(z(v_{x,\tilde{x}})) = \P\left[\E_{\hat{Q}}\left[L_{x,\tilde{x}}^{Y^{d}}\right]\leq v_{x,\tilde{x}}\Big \lvert X=x\right].
	\end{equation}
	Consider first the case when $v_{x,\tilde{x}}\geq L_{x,\tilde{x}}({P_{Y|x}}) + \delta$. Note that $L_{x,\tilde{x}}({P_{Y|x}}) = D(P_{Y|x}||P_{Y|\tilde{x}})>0$. We split this case into two subcases.
	\begin{itemize}
		\item If $v_{x,\tilde{x}}\geq\max_{b\in{\cal Y}}L_{x,\tilde{x}}^b$ then $\Gamma_{x,\tilde{x}}(z(v_{x,\tilde{x}})) = 1$. 
		\item Otherwise, consider the case where $L_{x,\tilde{x}}(P_{Y\mid x})+\delta\leq v_{x,\tilde{x}}\leq\max_{b\in{\cal Y}}L_{x,\tilde{x}}^b$. Now, observe that
		\begin{align}
			1\geq \Gamma_{x,\tilde{x}}(&z(v_{x,\tilde{x}}))\geq \P \left[\E_{\hat{Q}}\left[L_{x,\tilde{x}}^{Y^{d}}\right]\leq L_{x,\tilde{x}}(P_{Y\mid x})+\delta\Big \lvert X=x\right]. 
		\end{align}		
		Furthermore, by a simple application of Markov's inequality, we get that since $L_{x,\tilde{x}}(P_{Y\mid x}) = \E_{Y^d\sim P_{Y^d|x}}[L_{x,\tilde{x}}^{Y^d}]$, the following inequality holds:
		\begin{equation}
			\P \left[\E_{\hat{Q}}\left[L_{x,\tilde{x}}^{Y^{d}}\right]\leq L_{x,\tilde{x}}(P_{Y\mid x})+\delta\Big \lvert X=x\right]\geq \frac{\delta}{L_{x,\tilde{x}}(P_{Y\mid x})+\delta}.
		\end{equation}
		Hence, we have that $1\geq \Gamma_{x,\tilde{x}}(z(v_{x,\tilde{x}})) \geq \frac{\delta}{L_{x,\tilde{x}}(P_{Y\mid x})+\delta}$.
	\end{itemize}
	Hence, overall, if $v_{x,\tilde{x}}\geq L_{x,\tilde{x}}({P_{Y|x}}) + \delta$, we get that $\Gamma_{x,\tilde{x}}(z(v_{x,\tilde{x}})) = \Theta(1)$.
	
	Next, consider the case when $v_{x,\tilde{x}}\leq L_{x,\tilde{x}}({P_{Y|x}}) - \delta$. We again split this case into two subcases.
	\begin{itemize}
		\item If $v_{x,\tilde{x}}\leq\min_{b\in{\cal Y}}L_{x,\tilde{x}}^b$, then $\Gamma_{x,\tilde{x}}(z(v_{x,\tilde{x}})) = 0$.
		\item If $\min_{b\in{\cal Y}}L_{x,\tilde{x}}^b\leq v_{x,\tilde{x}}\leq L_{x,\tilde{x}}(P_{Y\mid X=x})$,
		then this is a large-deviations event. The probability $\Gamma_{x,\tilde{x}}(z(v_{x,\tilde{x}}))$ thus decays
		exponentially in $d$, i.e., $\Gamma_{x,\tilde{x}}(z(v_{x,\tilde{x}})) = e^{-dE_1(v_{x,\tilde{x}})+\Theta(\log d)}$, for some $E_1:\mathbb{R}\to [0,\infty)$, where by Sanov's theorem (see \cite[Sec. 11.4]{cover_thomas}), we have that \begin{equation}E_1(v_{x,\tilde{x}}) = \min_{Q\colon L_{x,\tilde{x}}(Q)\leq v_{x,\tilde{x}}}D(Q\mid\mid P_{Y\mid X=x}).\end{equation} We now aim to characterize this function $E_1$.
		
		From convex duality (see \cite[Sec. 5.2]{boyd}), we obtain that 
		\begin{equation}
			\label{eq:E1}
			E_1(v_{x,\tilde{x}}) = \max_{\lambda\geq0}\min_{Q}D(Q\mid\mid P_{Y\mid X=x})+\lambda(L_{x,\tilde{x}}(Q)-v_{x,\tilde{x}}).
		\end{equation}
		First, let us solve the inner minimization. Fix $\lambda\geq 0$. Observe that the term $f(\lambda,Q):= D(Q\mid\mid P_{Y\mid X=x}) + \lambda L_{x,\tilde{x}}(Q)$ obeys
		\begin{align}
			f(\lambda,Q)&=\sum_{y}Q(y)\left(\log\frac{Q(y)}{P_{Y\mid X=x}(y)}+\lambda L_{x,\tilde{x}}(y)\right)\\
			& =-\sum_{y}Q(y)\left(\log\frac{P_{Y\mid X=x}(y)}{Q(y)e^{\lambda L_{x,\tilde{x}}(y)}}\right)\\ &\geq-\log\sum_{y}\frac{P_{Y\mid X=x}(y)}{e^{\lambda L_{x,\tilde{x}}(y)}},
		\end{align}		
		using Jensen's inequality. Also note that equality holds when \begin{equation}Q(y) = Q^*(y)=\frac{1}{Z_{x,\tilde{x}}(\lambda)}P_{Y\mid X=x}(y)e^{-\lambda L_{x,\tilde{x}}(y)},\end{equation} where $Z_{x,\tilde{x}}(\lambda) := \sum_{y\in \mathcal{Y}} P_{Y\mid X=x}(y)e^{-\lambda L_{x,\tilde{x}}(y)}$ is the partition function (normalization constant). Substituting $Q^*$ back into the function results 
		\begin{align}
			f(\lambda,Q^*)&=-\log Z_{x,\tilde{x}}(\lambda) -\lambda v_{x,\tilde{x}}\\
			& =-\log\left(\sum_{y}P_{Y\mid X=x}(y)e^{-\lambda L_{x,\tilde{x}}(y)}\right)-\lambda v_{x,\tilde{x}}\\
			& =-\log\left(\sum_{y}P_{Y\mid X=x}^{1-\lambda}(y\mid x)P_{Y\mid X}^{\lambda}(y\mid\tilde{x})\right)-\lambda v_{x,\tilde{x}}.
		\end{align}		
		Plugging back into \eqref{eq:E1}, we get that 
		\begin{align}
			E_1(v_{x,\tilde{x}}) &= \max_{\lambda\geq0} \Bigg[-\log\left(\sum_{y}P_{Y\mid X=x}^{1-\lambda}(y\mid x)P_{Y\mid X}^{\lambda}(y\mid\tilde{x})\right) -\lambda v_{x,\tilde{x}}\Bigg].
		\end{align}		
	\end{itemize}
	Note that by setting $E(v_{x,\tilde{x}}) = E_1(v_{x,\tilde{x}})$ when $\min_{b\in{\cal Y}}L_{x,\tilde{x}}^b\leq v_{x,\tilde{x}}\leq L_{x,\tilde{x}}(P_{Y\mid X=x})$ and $E(v_{x,\tilde{x}}) = \infty$ when $v_{x,\tilde{x}}\leq\min_{b\in{\cal Y}}L_{x,\tilde{x}}^b$, we obtain a tight characterization of the rate of exponential decay of $\Gamma_{x,\tilde{x}}(z(v_{x,\tilde{x}}))$, in these ranges of $v_{x,\tilde{x}}$.
\end{proof}
We now state a prove a lemma that will allow us to identify the exact rate of exponential decay of $\mathsf{H}(X|Y^d)$. 
\begin{lemma}
	\label{lem:inter2}
	We have that
	\begin{equation}
		\int_{0}^{\infty} \max_{\tilde{x}\neq x} \Gamma_{x,\tilde{x}}(z) \d z = e^{\Theta(\log d|\mathcal{X}|)}\cdot e^{-d\cdot \min_{\tilde{x}\neq x} \mathsf{C}(P_{Y|x},P_{Y|\tilde{x}})}.
	\end{equation}
\end{lemma}
\begin{proof}
	Consider an integral of the form
	\begin{equation}
		\overline{I} = \int_{0}^{\infty} \max_{\tilde{x}\neq x} \Gamma_{x,\tilde{x}}(z) \d z.
	\end{equation}
	First, observe that 
	\begin{align}
		\overline{I} &\leq \int_{0}^{\infty} \sum_{\tilde{x}\neq x}\Gamma_{x,\tilde{x}}(z) \d z\\
		&\leq |\mathcal{X}|^2\cdot \max_{\tilde{x}\neq x}\int_{0}^{\infty} \Gamma_{x,\tilde{x}}(z) \d z\\
		&=e^{\Theta(\log |\mathcal{X}|)}\cdot \max_{\tilde{x}\neq x}\int_{0}^{\infty} \Gamma_{x,\tilde{x}}(z) \d z.
	\end{align}		
	Further,
	\begin{align}
		\overline{I} &\geq \max_{\tilde{x}\neq x}\int_{0}^{\infty} \Gamma_{x,\tilde{x}}(z) \d z.
	\end{align}	
	Hence, overall, we have that $\overline{I} = \text{exp}(\Theta(\log|\mathcal{X}|))\cdot \max_{\tilde{x}\neq x}\int_{0}^{\infty} \Gamma_{x,\tilde{x}}(z) \d z$. We next focus on the integral
	\begin{equation}
		I = \int_{0}^{\infty} \Gamma_{x,\tilde{x}}(z) \d z,
	\end{equation}
	for a fixed $x, \tilde{x}$. 
	
	By the definition \begin{equation}v_{x,\tilde{x}} = v_{x,\tilde{x}}(z) = \frac{1}{d}\left(\log\left(\frac{ce^{-z}}{1-e^{-z}}\right) - \log \frac{P(x)}{P(x')}\right),\end{equation} and from Prop. \ref{prop:Gamma}, we see that the integral $I$ can be split into two integrals for its evaluation: integral $I_1$, when \begin{equation}z\leq \log \left(1+\frac{cP(\tilde{x})}{P(x)}\cdot e^{-d(L_{x,\tilde{x}}(P_{Y|x})+\delta)}\right):=u(d),\end{equation} or equivalently, when $v_{x,\tilde{x}}\geq  L_{x,\tilde{x}}({P_{Y|x}}) + \delta$, and integral $I_2$, when \begin{equation}z\geq \log \left(1+\frac{cP(\tilde{x})}{P(x)}\cdot e^{-d(L_{x,\tilde{x}}(P_{Y|x})-\delta)}\right):=\ell(d),\end{equation} or equivalently, when $v_{x,\tilde{x}}\leq  L_{x,\tilde{x}}({P_{Y|x}}) - \delta$, with $I = I_1+I_2$. Here, we pick $\delta< L_{x,\tilde{x}}({P_{Y|x}})$.
	
	Consider first the integral $I_1$. We have that
	\begin{align}
		I_1 &= \int_{z\leq u(d)}  \Gamma_{x,\tilde{x}}(z) \d z\\
		&= \Theta(1)\cdot  u(d).
	\end{align}		
	Now, we claim that $u(d)= e^{-d(L_{x,\tilde{x}}(P_{Y|x})+\delta)}$. To see this, let $\tilde{c}:= \frac{cP(\tilde{x})}{P(x)}$ and recall that $\frac{t}{1+t}\leq \log(1+t)\leq t$, for $t>-1$. Hence, we have that
	\begin{equation}
		u(d)\leq \tilde{c}e^{-d(L_{x,\tilde{x}}(P_{Y|x})+\delta)},
	\end{equation}
	and
	\begin{equation}
		u(d)\geq  \frac{\tilde{c}e^{-d(L_{x,\tilde{x}}(P_{Y|x})+\delta)}}{1+\tilde{c}e^{-d(L_{x,\tilde{x}}(P_{Y|x})+\delta)}} \geq \frac{\tilde{c}}{2}\cdot e^{-d(L_{x,\tilde{x}}(P_{Y|x})+\delta)},
	\end{equation}
	since $L_{x,\tilde{x}}({P_{Y|x}}) = D(P_{Y|x}||P_{Y|\tilde{x}})>0$. 
	Hence, overall, we have that for any fixed $\delta>0$, \begin{equation}I_1 =  e^{-d(L_{x,\tilde{x}}(P_{Y|x})+\delta)+\Theta(1)}.\end{equation} Taking the limit as $\delta \downarrow 0$, we get that 
	\begin{equation}I_1=  e^{-dL_{x,\tilde{x}}(P_{Y|x})+\Theta(1)} =  e^{-d\cdot D(P_{Y|x}||P_{Y|\tilde{x}})+\Theta(1)}.  \label{eq:i1} \end{equation}
	
	Consider now the integral $I_2$. Note that from case 2) in Prop. \ref{prop:Gamma}, when $z\geq  \log \left(1+\tilde{c}e^{-d\cdot \min_b L_{x,\tilde{x}}^b}\right):=U(d)$, we have that $I_2 = 0$. Hence, we can write
	\begin{equation}
		I_2 = \int_{\ell(d)\leq z\leq U(d)}  \Gamma_{x,\tilde{x}}(z) \d z,
	\end{equation}
	where, for this interval $\ell(d)\leq z\leq U(d)$ of interest, 
	\begin{equation}
		\Gamma_{x,\tilde{x}}(z) = e^{-d{E}(v_{x,\tilde{x}}(z))+\Theta(\log d)},
	\end{equation}
	with \begin{equation}{E}(v_{x,\tilde{x}}(z)) = \max_{\lambda\geq 0}\left[-\log \left(\sum_{y} P(y|x)^{1-\lambda}P(y|x')^\lambda\right)-\lambda v_{x,\tilde{x}}\right],\label{eq:E}\end{equation} from Prop. \ref{prop:Gamma}.
	
	It is clear that $E(v_{x,\tilde{x}}(z))$ increases as $v_{x,\tilde{x}}$ decreases. In the integral $I_2$, we thus make the change of variables
	\begin{equation}
		v = \frac{1}{d}\left(\log\left(\frac{ce^{-z}}{1-e^{-z}}\right) - \log \frac{P(x)}{P(x')}\right),
	\end{equation}
	with $\d z = \frac{-dc e^{-dv}}{1+ce^{-dv}} \d v$. Thus,
	\begin{align}
		I_2 &=  dc\int^{L_{x,\tilde{x}}(P_{Y|x})-\delta}_{v=\min_b L_{x,\tilde{x}}^b} e^{-dE(v)} \cdot \frac{e^{-dv}}{1+ce^{-dv}} \d v\\
		&= dc\int^{L_{x,\tilde{x}}(P_{Y|x})-\delta}_{v=0} e^{-dE(v)}\cdot   \frac{e^{-dv}}{1+ce^{-dv}} \d v +  dc\int^{0}_{v=\min_b L_{x,\tilde{x}}^b} e^{-dE(v)}\cdot \frac{e^{-dv}}{1+ce^{-dv}} \d v.
	\end{align}	
	Let us call the first integral in the expansion above as $I_{2,a}$ and the second as $I_{2,b}$ (note that $\min_b L_{x,\tilde{x}}^b\leq 0$). Now, observe that 
	\begin{align}
		I_{2,a}&=dc\int^{L_{x,\tilde{x}}(P_{Y|x})-\delta}_{v=0} e^{-dE(v)}\cdot {e^{-dv+\Theta(\log d)}} dv\\
		&= \text{exp}(-d(E(L_{x,\tilde{x}}(P_{Y|x})) + L_{x,\tilde{x}}(P_{Y|x})-\delta)+\Theta(\log d)).
	\end{align}	
	Here, in the first eqauility, we argue as previously that 
	\begin{equation}
		\frac{e^{-dv}}{1+ce^{-dv}} = e^{-dv+\Theta(1)}.
	\end{equation}
	Furthermore, the second equation follows from Laplace's method (e.g., see \cite[Sec. 4.2]{nerifnt}), using the fact that both $e^{-dE(v)}$ and $e^{-dv}$ increase as $v$ increases. Taking $\delta\downarrow 0$, we obtain that
	\begin{equation}
		I_{2,a}= e^{-d(E(L_{x,\tilde{x}}(P_{Y|x})) + L_{x,\tilde{x}}(P_{Y|x}))+\Theta(\log d)}.
		\label{eq:i2a}
	\end{equation}
	
	Moreover,
	\begin{align}
		I_{2,b} &=e^{\Theta(\log d)}\cdot \int^{0}_{v=\min_b L_{x,\tilde{x}}^b}e^{-dE(v)} \d v  \\
		&= e^{-dE(0)+\Theta(\log d)} = e^{-d\cdot \mathsf{C}(P_{Y|x},P_{Y|\tilde{x}})+\Theta(\log d)}. \label{eq:i2b}
	\end{align}
	Here, for the first equality, we use the fact that when $v\leq 0$, 
	\begin{equation}
		\frac12\leq \frac{e^{-dv}}{1+ce^{-dv}} \leq \frac{1}{c},
	\end{equation}
	and hence $\frac{e^{-dv}}{1+ce^{-dv}} = \Theta(1)$.
	The second equality holds again by Laplace's method, using the fact that $e^{-dE(v)}$ increases with $v$.
	
	Putting everything together, in order to nail down the exponential behaviour of integral $I$, we need to reconcile Equations \eqref{eq:i1}, \eqref{eq:i2a}, and \eqref{eq:i2b}, to understand which among $I_1$, $I_{2,a}$, and $I_{2,b}$, is the largest. First, observe that since \begin{equation}\mathsf{C}(P_{Y|x},P_{Y|\tilde{x}})\leq D(P_{Y|x}||P_{Y|\tilde{x}}),\end{equation} we have that $I_{2,b}$ dominates $I_1$.
	
	Therefore, the exponential behaviour of integral $I$ is governed by the rate
	\begin{equation}
		\overline{E} = \min \left\{E(L_{x,\tilde{x}}(P_{Y|x})) + L_{x,\tilde{x}}(P_{Y|x}), \mathsf{C}(P_{Y|x},P_{Y|\tilde{x}}) \right\}.
	\end{equation}
	Now note that the maximum over $\lambda \geq 0$ in the expressions for $E$ and $\mathsf{C}$ can be rewritten as a maximum over $\lambda \in [0,1]$ (owing to the concavity of $-\log \left(\sum_{y} P(y|x)^{1-\lambda}P(y|x')^\lambda\right)$ and the fact that it equals $0$ at $\lambda = 0,1$). We then get that
	\begin{equation}
		E(L_{x,\tilde{x}}(P_{Y|x})) + L_{x,\tilde{x}}(P_{Y|x})\geq \mathsf{C}(P_{Y|x},P_{Y|\tilde{x}}),
	\end{equation}
	implying that $\overline{E} = \mathsf{C}(P_{Y|x},P_{Y|\tilde{x}})$. Therefore, \begin{equation}I= e^{-d\cdot \mathsf{C}(P_{Y|x},P_{Y|\tilde{x}})+\Theta(\log d)},\end{equation} and hence that \begin{equation}\overline{I}= \max_{\tilde{x}\neq x} \text{exp}({-d\cdot \mathsf{C}(P_{Y|x},P_{Y|\tilde{x}})+\Theta(\log d)+\Theta(\log |\mathcal{X}|)}).\end{equation} 
\end{proof}
We are now ready to prove Theorem \ref{thm:conv} for the convergence of the mutual information.
\begin{proof}[Proof of Thm. \ref{thm:conv} for mutual information]
	Recall from Lemma \ref{lem:inter} that
	\begin{equation}
		p_x(t)\leq (|\mathcal{X}|-1)\cdot\max_{\tilde{x}\neq x}\Gamma_{x,\tilde{x}}(t),
	\end{equation}
	with the constant $c$ in the definition of $\Gamma_{x,\tilde{x}}(t)$ taken to be $|\mathcal{X}|-1$. Thus, from Lemma \ref{lem:inter2} and \eqref{eq:inter3}, we get that
	\begin{align}
		\E\left[\log\frac{1}{P_{X\mid Y^{d}}(x\mid Y^{d})}\Bigg \lvert X=x\right] &= \int_{0}^\infty p_x(t)\d t\\ &\leq \text{exp}\left({-d\cdot \min_{\tilde{x}\neq x} \mathsf{C}(P_{Y|x},P_{Y|\tilde{x}})+\Theta(\log d|\mathcal{X}|)}\right),
	\end{align}	
	where the inequality holds by the Laplace method  \cite[Sec. 4.2]{nerifnt}.
	Similarly, since
	\begin{equation}
		p_x(t)\geq \max_{\tilde{x}\neq x} \Gamma_{x,\tilde{x}}(t),
	\end{equation}
	with the constant $c$ in the definition of $\Gamma_{x,\tilde{x}}(t)$ taken to be $1$, we get that
	\begin{align}
		\E\left[\log\frac{1}{P_{X\mid Y^{d}}(x\mid Y^{d})}\Bigg \lvert X=x\right] &\geq \text{exp}\left({-d\cdot \min_{\tilde{x}\neq x} \mathsf{C}(P_{Y|x},P_{Y|\tilde{x}})+\Theta(\log d|\mathcal{X}|)}\right).
	\end{align}	
	Thus,
	\begin{align}
		\E\left[\log\frac{1}{P_{X\mid Y^{d}}(x\mid Y^{d})}\Bigg \lvert X=x\right]&= \text{exp}\left({-d\cdot \min_{\tilde{x}\neq x} \mathsf{C}(P_{Y|x},P_{Y|\tilde{x}})+\Theta(\log d|\mathcal{X}|)}\right).
	\end{align}	
	Plugging back in \eqref{eq:inter2}, we obtain that \begin{equation}\mathsf{H}(X|Y^d)= \max_{\substack{x,\tilde{x}: \tilde{x}\neq x}} \text{exp}\left({-d\cdot \mathsf{C}(P_{Y|x},P_{Y|\tilde{x}})+\Theta(\log d|\mathcal{X}|)}\right),\end{equation} yielding the required result.
\end{proof}

\subsection{Convergence of Channel Dispersion of DMCs}
\label{sec:dispersion}

We now shift our attention to proving the convergence of the dispersion $V^{(d)} = \mathbb{E}\left[\left(\iota(X;Y^d) - I(X;Y^d)\right)^2\right]$, to the varentropy $\mathsf{V}(X)$ of the input distribution $P_X$, for general DMCs $W$. 

First, we write
\begin{align}
	V^{(d)} &= \E\left[\left(\log \frac{P(X|Y^d)}{P(X)} - \left(\mathsf{H}(X) - \mathsf{H}(X|Y^d)\right)\right)^2\right] \\
	&= \E\left[\left(\log \frac{1}{P(X)} - \mathsf{H}(X)\right)^2\right] + \E\left[\left(\log P(X|Y^d) + \mathsf{H}(X|Y^d)\right)^2\right] \nonumber \\
	&\ \ \ + 2\cdot\E\left[(\log P(X|Y^d)+\mathsf{H}(X|Y^d))\cdot \left(\log \frac{1}{P(X)} - \mathsf{H}(X)\right)\right] \\
	&= \mathsf{V}(X)+\left(\E\left[\left(\log P(X|Y^d)\right)^2\right] - \mathsf{H}(X|Y^d)^2\right)+\theta_d, \label{eq:interdisp}
\end{align}
where we use the fact that $\E\left[-\log P(X|Y^d)\right] = \mathsf{H}(X|Y^d)$, and the cross-term $\theta_d$ is
\begin{equation}
	2\cdot\E\left[(\log P(X|Y^d)+\mathsf{H}(X|Y^d))\cdot \left(\log \frac{1}{P(X)} - \mathsf{H}(X)\right)\right].
\end{equation}

Further, observe that when $P_X = \text{Unif}(\mathcal{X})$, the cross term vanishes, since for any $x\in \mathcal{X}$, we have that $\log \frac{1}{P(x)} = \log |\mathcal{X}| = \mathsf{H}(X)$. Our proof strategy hence is to first characterize (up to exponential tightness) the behaviour of $\E\left[\left(\log P(X|Y^d)\right)^2\right]$ and then handle the cross-term $\theta_d$.
\begin{lemma}
	\label{lem:interdisp1}
	We have that
	\begin{equation}
		\E\left[\left(\log P(X|Y^d)\right)^2\right]= e^{\Theta(\log d|\mathcal{X}|)}\cdot \max_{x,\tilde{x}: \tilde{x}\neq x} e^{-d\cdot \mathsf{C}(P_{Y|x},P_{Y|\tilde{x}})}.
	\end{equation}
\end{lemma}
\begin{proof}
	This proof is entirely analogous to the proof of convergence of $I^{(d)}$ to $\mathsf{H}(X)$, for general DMCs, in Section \ref{sec:gendmccap}, and we only highlight the key steps.
	
	We first consider the term $\E\left[\left(\log P(X|Y^d)\right)^2\right]$. As in \eqref{eq:inter2}, we first expand this term as:
	\begin{equation}
		\E\left[\left(\log P(X|Y^d)\right)^2\right] = \sum_{x\in \mathcal{X}}P(x)\E\left[\left(\log P(X|Y^d)\right)^2\Bigg \lvert X=x\right],
	\end{equation}
	where for any $x\in \mathcal{X}$, we have
	\begin{align}
		\E\left[\left(\log P(x|Y^d)\right)^2\Bigg \lvert X=x\right]&= \int_{t=0}^\infty \Pr\left[\left(\log P(x|Y^d)\right)^2 \geq t \Bigg \lvert X=x\right] \d t\\
		&= \int_{t=0}^\infty \P\left[\log P(x|Y^d)\geq \sqrt{t}\Big \lvert X=x\right] \d t +  \int_{t=0}^\infty \P\left[\log P(x|Y^d)\leq -\sqrt{t}\Big \lvert X=x\right] \d t\\
		&= \int_{t=0}^\infty \P\left[\log P(x|Y^d)\leq -\sqrt{t}\Big \lvert X=x\right] \d t.
	\end{align}		
	Here, the last equality holds since $P(x|Y^d)\leq 1$, for any $Y^d\in \mathcal{Y}^d$. As earlier, our task reduces to understanding the exponential behaviour of \begin{equation}q_x(t):= \Pr\left[\log P(x|Y^d)\leq -\sqrt{t}\Big \lvert X=x\right].\end{equation}
	
	Exactly as in Lemma \ref{lem:inter}, we have that
	\begin{align}
		\frac{q_x(t)}{(|{\cal X}|-1)}&\leq\max_{\tilde{x}\neq x}\P\left[\frac{P(x)P(Y^{d}\mid x)}{P(\tilde{x})P(Y^{d}\mid\tilde{x})}\leq \frac{e^{-(\sqrt{t}-\log(|{\cal X}|-1))}}{1-e^{-\sqrt{t}}}\Bigg \lvert X=x\right],
    \end{align}
and 
    \begin{equation}
    q_x(t)	\geq \max_{\tilde{x}\neq x}\P\left[\frac{P(x)P(Y^{d}\mid x)}{P(\tilde{x})P(Y^{d}\mid\tilde{x})}\leq\frac{e^{-\sqrt{t}}}{1-e^{-\sqrt{t}}}\Bigg \lvert X=x\right].
	\end{equation}	
	Consider, for a fixed $c>0$ and $z\in [0,\infty)$, the expression
	\begin{equation}
		\Delta_{x,\tilde{x}}(z):= \P\left[\frac{P(x)P(Y^{d}\mid x)}{P(\tilde{x})P(Y^{d}\mid\tilde{x})}\leq \frac{ce^{-\sqrt{z}}}{1-e^{-\sqrt{z}}}\Bigg \lvert X=x\right].
	\end{equation}
	As argued in the case of the mutual information, the rate of exponential decay of $\Delta_{x,\tilde{x}}(z)$, as a function of $z$, dictates the rate of exponential decay of $q_x(t)$, and hence of $\E\left[\left(\log P(X|Y^d)\right)^2\right]$. Once again, letting \begin{equation}v_{x,\tilde{x}} = v_{x,\tilde{x}}(z) = \frac{1}{d}\left(\log\left(\frac{ce^{-z}}{1-e^{-z}}\right) - \log \frac{P(x)}{P(x')}\right)\end{equation} and setting $(z(v_{x,\tilde{x}}))^{1/2}:= \log\left(1+\frac{cP(\tilde{x})}{P(x)\cdot }e^{-dv_{x,\tilde{x}}}\right)$, we obtain a result similar to Proposition \ref{prop:Gamma}, with the only modification being the replacement of $\Gamma_{x,\tilde{x}}$ by $\Delta_{x,\tilde{x}}$.
	
	Next, as in Lemma \ref{lem:inter2}, consider the integral
	\begin{equation}
		\overline{I}^{V} = \int_{0}^{\infty} \max_{\tilde{x}\neq x} \Delta_{x,\tilde{x}}(z) \d z.
	\end{equation}
	We claim that \begin{equation}\overline{I}^{V} = e^{\Theta(\log d)+\Theta(\log |\mathcal{X}|)}\cdot e^{-d \min_{\tilde{x}\neq x} \mathsf{C}(P_{Y|x},P_{Y|\tilde{x}})}.\end{equation} To obtain this, we argue as in the proof of Lemma \ref{lem:inter2} that $\overline{I}^{V} = e^{\Theta(\log |\mathcal{X}|)}\cdot \max_{\tilde{x}\neq x} I^V$, where
	\begin{equation}
		I^V = \int_{0}^{\infty} \Delta_{x,\tilde{x}}(z) \d z,
	\end{equation}
	for a fixed $x, \tilde{x}$.
	
	Now, we split the evaluation of $I^V$ into two integrals: integral $I_1^V$, when \begin{equation}z\leq \left(\log \left(1+\frac{cP(\tilde{x})}{P(x)}\cdot e^{-d(L_{x,\tilde{x}}(P_{Y|x})+\delta)}\right)\right)^2,\end{equation} and integral $I_2^V$, when \begin{equation}z\geq \left(\log \left(1+\frac{cP(\tilde{x})}{P(x)}\cdot e^{-d(L_{x,\tilde{x}}(P_{Y|x})-\delta)}\right)\right)^2,\end{equation} with $I^V = I_1^V+I_2^V$. By arguments analogous to those in the proof of Lemma \ref{lem:inter2}, we obtain that 
	\begin{equation}
		I_1^V= e^{-2d\cdot D(P_{Y|x}||P_{Y|\tilde{x}})+\Theta(\log d)},
	\end{equation}
	and
	\begin{equation}
		I_2^V= e^{-d\cdot \mathsf{C}(P_{Y|x},P_{Y|\tilde{x}})+\Theta(\log d)}.
	\end{equation}
	Thus, we have that \begin{equation}I^V = e^{-d\cdot \mathsf{C}(P_{Y|x},P_{Y|\tilde{x}})+\Theta(\log d)},\end{equation} and hence that \begin{equation}\overline{I}^V = \text{exp}\left({-d\cdot \min_{\tilde{x}\neq x} \mathsf{C}(P_{Y|x},P_{Y|\tilde{x}})+\Theta(\log d|\mathcal{X}|)}\right).\end{equation} Putting everything together, we get that 
    \begin{align}
     \E\left[\left(\log P(X|Y^d)\right)^2\right]
     &= \max_{x,\tilde{x}: \tilde{x}\neq x} \text{exp}\left({-d\cdot  \mathsf{C}(P_{Y|x},P_{Y|\tilde{x}})+\Theta(\log d|\mathcal{X}|)}\right).
    \end{align}
\end{proof}

\begin{lemma}
	\label{lem:interdisp2}
	We have that
	\begin{equation}
		\theta_d=  e^{\Theta(\log d|\mathcal{X}|)}\cdot \max_{x,\tilde{x}: \tilde{x}\neq x} e^{-d\cdot \mathsf{C}(P_{Y|x},P_{Y|\tilde{x}})}.
	\end{equation}
\end{lemma}
\begin{proof}
	Recall the cross term 
 \begin{equation}
 \theta_d:= 2\cdot \E\left[(\log P(X|Y^d)+\mathsf{H}(X|Y^d))\cdot \left(\log \frac{1}{P(X)} - \mathsf{H}(X)\right)\right].
 \end{equation}
 
 We then write
	\begin{align}
		\theta_d(x) &:= \left(\log \frac{1}{P(x)} - \mathsf{H}(X)\right)\cdot \left(\E\left[\log P(x|Y^d)\Big \lvert X=x\right]+\mathsf{H}(X|Y^d)\right),
	\end{align}
	with $\theta_d = \sum_{x\in \mathcal{X}}P(x)\theta_d(x)$. From the capacity analysis, we know that (see \eqref{eq:inter3}) 
	\begin{align}
    \E\left[\log P(x|Y^d)\Big \lvert X=x\right] 
    &= -\max_{x'\neq x}\text{exp}\left({-d\cdot  \mathsf{C}(P_{Y|x},P_{Y|\tilde{x}})+\Theta(\log d|\mathcal{X}|)}\right). 
	\end{align}	
    From our previous analysis, we know that 
    \begin{equation}\mathsf{H}(X|Y^d)= \max_{x,x': x'\neq x} \text{exp}\left({-d\cdot  \mathsf{C}(P_{Y|x},P_{Y|\tilde{x}})+\Theta(\log d|\mathcal{X}|)}\right).
    \end{equation} 
	Hence,
	\begin{align}
    \theta_d&= e^{\Theta(\log d)}\cdot\theta_d(x) = \max_{x,x': x\neq x'}\text{exp}\left({-d\cdot  \mathsf{C}(P_{Y|x},P_{Y|\tilde{x}})+\Theta(\log d|\mathcal{X}|)}\right).
	\end{align}	
\end{proof}
With these lemmas in place, the proof of Theorem is easily completed.
\begin{proof}[Proof of Thm. \ref{thm:conv} for dispersion]
	From \eqref{eq:interdisp} we see that
	\begin{equation}
		V^{(d)}-\mathsf{V}(X)= \E\left[\left(\log P(X|Y^d)\right)^2\right] - \mathsf{H}(X|Y^d)^2+\theta_d.
	\end{equation}
	Putting together Lemmas \ref{lem:interdisp1} and \ref{lem:interdisp2} and the fact that \begin{equation}\mathsf{H}(X|Y^d)= \max_{\substack{x,x': x'\neq x}} \text{exp}\left({-d\cdot \mathsf{C}(P_{Y|x},P_{Y|\tilde{x}})+\Theta(\log d|\mathcal{X}|)}\right)\end{equation} (from the proof of Thm. \ref{thm:conv} for mutual information), we readily obtain the required result.
\end{proof}
\section{Conclusion and Future Work}
\label{sec:conclusion}
In this paper, we considered the behaviour of the information rate and channel dispersion of  multi-view noisy channels, with arbitrary, but fixed, input distributions, in the regime where the number of views is very large. In the first part of the paper, we considered multi-view discrete memoryless channels (DMCs), and showed that the information rate and the channel dispersion converge exponentially quickly to the input entropy and varentropy, respectively, with an input distribution-independent rate that equals the smallest Chernoff information between two conditional distributions of the output given different inputs. We then extended these results to general multi-letter channels, and computed explicit upper bounds on the Chernoff information above for the deletion channel. Finally, for the special case of the multi-view binary symmetric channel, we also presented a close, non-asymptotic approximation to its capacity, via the capacity of the so-called Poisson approximation channel.

One interesting direction for future work would be the extension of the techniques and results in this paper to other multi-view channels with memory and general synchronization channels (see \cite{haeupler} and \cite[Sec. VII]{mahdi_survey}).
\ifCLASSOPTIONcaptionsoff
  \newpage
\fi



%
\bibliographystyle{IEEEtran}
{\footnotesize
	\bibliography{references}}

%

%





\end{document}